\documentclass[a4paper,12pt]{iopart}
\usepackage{amsmath_iop}
\usepackage[utf8]{inputenc}
\usepackage[english]{babel}
\usepackage{hyperref}
\usepackage{tikz}
\usepackage{amstext,amsthm,amssymb}
\usepackage{dsfont}
\usepackage{enumitem}
\usepackage{times}
\makeatletter
\newtheorem*{rep@theorem}{\rep@title}
\newcommand{\newreptheorem}[2]{\newenvironment{rep#1}[1]{\def\rep@title{#2 \ref{##1}}
 \begin{rep@theorem}}
 {\end{rep@theorem}}}
\makeatother

\newtheorem{theorem}{Theorem}
\newreptheorem{theorem}{Theorem}
\newtheorem{lemma}{Lemma}
\newreptheorem{lemma}{Lemma}

\newtheorem{proposition}{Proposition}
\newreptheorem{proposition}{Proposition}
\DeclareMathOperator{\sgn}{sgn}
\newcommand{\del}{\partial}
\newcommand{\bra}[1]{\ensuremath{\langle #1|}}
\newcommand{\ket}[1]{\ensuremath{|#1\rangle}}
\newcommand{\braket}[2]{\ensuremath{\langle #1|#2\rangle}}

\newcommand{\mean}[1]{\ensuremath{\langle{#1}\rangle}}

\newcommand{\C}{\ensuremath{\mathbb{C}}}

\newcommand{\id}{\mathds{1}}
\newcommand{\N}{\ensuremath{\mathbb{N}}}
\newcommand{\Z}{\ensuremath{\mathds{Z}}}
\renewcommand{\tilde}{\widetilde}

\begin{document}

\title{Symmetric hypergraph states: 
Entanglement quantification and robust Bell nonlocality
}
\author{Jan N\"{o}ller$^{1,\dagger}$, Otfried G\"uhne$^{2}$, and Mariami Gachechiladze$^{1,2}$}
\address{$^1$Department of Computer Science, Technical University of Darmstadt, Germany}
\address{$^2$Naturwissenschaftlich-Technische Fakult\"at,
Universit\"at Siegen, Germany}
\address{$^\dagger$Corresponding author}

\begin{abstract}
Quantum hypergraph states are the natural generalization of graph states. Here we investigate and analytically quantify entanglement and nonlocality for large classes of quantum hypergraph states. More specifically, we connect the geometric measure of entanglement of symmetric hypergraphs to their local Pauli stabilizers. As a result we recognize the resemblance between symmetric graph states and symmetric hypergraph states. This explains both the exponentially increasing violation of local realism for infinitely many classes of hypergraph states and robustness towards particle loss. 
\end{abstract}

\section{Introduction}
Multipartite entanglement is believed to be the key ingredient for many applications such as quantum simulation, metrology, and protocols in quantum information processing. Accordingly, its quantitative and qualitative characterization is of great importance. However, this task has turned out to be difficult due to the exponentially increasing dimension of the Hilbert space, where these states live. What is more, it is known that almost all quantum states in this huge Hilbert space are useless for quantum information processing~\cite{Hayden2006Aspects,Gross2009Most,Bremner2009Are}. 
Consequently, the research has focused on classes of multipartite entangled states, which are easy to characterize, manipulate and have a wide gamut of applications. In fact, symmetries and other kinds of simplifications seem to be essential for a state to be an interesting resource.

One of the fundamental ways to explore the structure of multipartite states is to quantify their entanglement using entanglement measures. In this work, we concentrate on geometric measure of entanglement~\cite{wei2003geometric,shimony1995degree}, which has become a staple method due to its desirable properties of an entanglement monotone. The geometric measure calculates the distance of a given state from the set of separable
pure states. Despite its simple definition, it is very hard to compute due to a large number of optimization parameters. The geometric measure  has been analytically estimated only for a few classes of states~\cite{wei2003geometric, markham2007entanglement,frydryszak2012geometric}, upper and lower bounds have been derived~\cite{qi2005eigenvalues,chen2010computation,hu2016computing,streltsov2011simple,xiong2022geometric}, and numerical methods have been considered to find states with high entanglement~\cite{qi2018entangled,steinberg2022maximizing,martin2010multiqubit}. In order to ease the complexity of the problem, symmetries of a quantum state have been utilized. It was proven that the closest separable state to a symmetric multiqubit state is also symmetric~\cite{Robert2009Geometric}. 


Hypergraph states~\cite{kruszynska2009local, qu2013encoding,rossi2013quantum, Guehne2014Hypergraph} are generalizations of graph states, which themselves are one of the most prominent classes of useful entangled states~\cite{hein2006entanglement,Raussendorf2001One,Raussendorf2003Measurement,Schlingemann2001Quantum,Looi2008Quantum,Hajdu2013Direct}, with their symmetric representative being the Greenberger-Horne-Zeilinger (GHZ) state. The geometric measure of GHZ states is trivial to calculate and equal to $1/2$ for any number of qubits. For graph states in general, however, only limited results are known~\cite{Markham2008Graph}, even though they are simple to describe and work with due to the stabilizer formalism. Hypergraph states have a much richer structure involving nonlocal stabilizers, which makes them on the one hand complicated to classify~\cite{Guehne2014Hypergraph,Gachechiladze2017Graphical,zhu2019efficient,morimae2017verification}, but on the other hand an interesting and robust resource for various information processing tasks~\cite{gachechiladze2016Extreme,gachechiladze2019chainging,takeuchi2019quantum}.
In this work, we investigate the geometric measure of entanglement of hypergraph states. We mainly focus on classes of fully  symmetric states and derive analytic expressions for them relying on their local symmetries, see Ref.~\cite{Lyons2017Local} for a full characterization of those. The study of (local) symmetries of hypergraph states is in itself an interesting subject of research. In this context, it is worthwhile to recall that extension of local complementation from graphs to hypergraphs was the central tool to construct a family of counterexamples to the famous LU=LC conjecture~\cite{ji2007lu,tsimakuridze2017graph}. Curiously, our proofs employ as a main tool the square-root of a stabilizer, which resembles the map used for local complementation.

We give methods to simplify calculations by reducing the number of parameters involved in the optimization by exploiting these symmetries. 
As a result we derive analytical expressions for the geometric measure. As a by-product of our method, we obtain an interesting way to write down hypergraph states, which explains the resemblance with GHZ states when considering their nonlocal properties. Moreover, we reproduce known reults about extreme nonlocality, i.e. exponential violation of Mermin-type inequalities in a much more concise and intuitive manner, and further generalize the robustness results to more states and qubits. These findings give a new insight into the structure of hypergraph states and could be used to derive new Bell inequalities and self-testing arguments.   

\begin{figure}[t]
  \centering
  \includegraphics[scale=0.6]{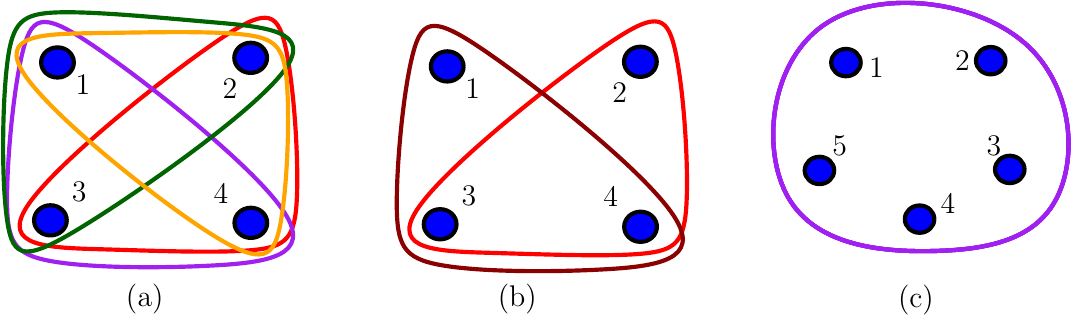}
  \caption{Quantum hypergraph states. (a) A hypergraph state containing all possible hyperedges consisting of three vertices each. This hypergraph state is called a \textit{complete $3$-uniform} hypergraph state. (b) A four-qubit hypergraph state. (c) A hypergraph state with a single hyperedge.}
  \label{fig:figure1}
\end{figure}

The paper is organized in the following way:
First, we review hypergraph states and the conditions under which complete symmetric hypergraph states are stabilized by local Pauli operators. We proceed by using these symmetries to find new representations of the respective hypergraph states under local unitary transformations. Subsequently, we exploit these representations to derive results of the geometric entanglement measure. Finally, we give proofs for the extreme violation of nonlocality in hypergraph states, and lastly, discuss robustness of entanglement in hypergraph states against particle loss.

\subsection*{Notation}
With $X,Y,Z$, we denote the three Pauli operators $\sigma_x, \sigma_y,\sigma_z$, and we use $P$ as placeholder if we want to refer to them collectively. 
In that spirit, $\sqrt{P}_{\pm}$ shall denote the operator which squares to $P\in\{X,Y,Z\}$ and has eigenvalues $1,\pm i$.
One class of states frequently appearing are different types of GHZ states, for which we use the following shorthand notation:
\[\ket{GHZ_P^{\pm}}=\frac{1}{\sqrt{2}}\left(\ket{+_P}^{\otimes N}\pm\ket{-_P}^{\otimes N}\right),\] with $\ket{\pm_P}$ being the $\pm 1$-eigenstate of $P$. Here, we usually omit the subscript when referring to the computational basis ($Z$) and do the same for the superscript if the relative phase is positive.
Many of our calculations will involve the \emph{weight} of a computational basis element $\ket{x}=\ket{i_1,\dots.i_N},\,i_j\in\{0,1\}$ which is just $w(x)=i_1+\dots+i_N$. 
\section{Hypergraph states}

Consider a hypergraph $H=(V,E)$, defined over a set of vertices $V$ and a set of hyperedges $E$, which may connect more than two vertices. From $H$ we can naturally construct a $|V|$-qubit quantum state, defined in the following manner:
   \begin{equation}
       \ket{H}=\prod_{e\in E} C_e \ket{+}^{\otimes |V|},
   \end{equation}
   where $C_e$ gates are generalized $CZ$ gates on $|e|$ qubits and are defined as $C_e=\id-2 \ket{11\dots 1}\bra{11\dots 1}$. See Fig.~\ref{fig:figure1} for some examples of hypergraphs. We say that a hypergraph state is $k$-uniform if all of its hyperedges connect exactly $k$ vertices. As an example, the hypergraphs in  Fig.~\ref{fig:figure1} (a) and (b) are $3$-uniform and Fig.~\ref{fig:figure1} (c) is a $5$-uniform hypergraph. Similarly to the graph notation, we say that a hypergraph state is $k$-uniform complete if it contains all hyperedges of size $k$.  Additionally, Fig.~\ref{fig:figure1} (a) is a $3$-uniform complete hypergraph and Fig.~\ref{fig:figure1} (c) is $5$-uniform complete. Complete hypergraphs correspond to the permutation symmetric states. Finally, we say that a state is $\textbf{k}$-uniform complete, for a vector $\textbf{k}=(k_1,\dots,k_m)$, if it contains all hyperedges of cardinality $k_i$ for all $i=1,\dots,m$.

Like for graph states, hypergraph states have an alternative definition using a stabilizer formalism. However, unlike in the graph state case, these stabilizers are nonlocal, i.e. they are not tensor products of local Pauli operators. Instead they contain nonlocal phase gates. For a vertex $i\in V$, the associated stabilizer operator is given by the expression:
\begin{equation}\label{eq:HGstabilizer}
    h_i=X_i \bigotimes_{e_j \in \mathcal{A}(i)}C_{e_j},
\end{equation}
where $X_i$ is Pauli-$X$ gate acting on $i$-th qubit and $\mathcal{A}(i)$ is the adjacency set of the vertex $i$, defined as $\mathcal{A}(i)=\{e - \{i\}|e \in E \mbox{ with } i \in e\}$. To put it simpler, the elements of $\mathcal{A}(i)$ are sets of vertices which are adjacent to $i$ via some hyperedge. The hypergraph state $\ket{H}$ is then the unique pure state which is invariant under the action of the group generated by those stabilizer operators.

However, there are cases where hypergraph states have local Pauli stabilizer. To give a simple example, the hypergraph state in Fig.~\ref{fig:figure1} (b) is an eigenstate of $X_1\otimes X_2$ operator and less trivially,  Fig.~\ref{fig:figure1} (a) is an eigenstate of $Y_1\otimes Y_2 \otimes Y_3 \otimes Y_4$. In Ref.~\cite{Lyons2017Local}, necessary and sufficient conditions were derived for symmetric hypergraph states to have local Pauli stabilizers and the explicit form of these stabilizers was given. For completeness we paraphrase this result here.

\begin{lemma}\label{lemma:localStab}
\cite{Lyons2017Local} A symmetric $N$-qubit $\textbf{k}$-uniform complete hypergraph state is
\begin{itemize}
\item[(1)]$+1$ - eigenstate of $X^{\otimes N}$ iff for $0\leq w\leq N$
\begin{equation}
\sum_{i=1}^m\binom{ w}{k_i}=\sum_{i=1}^m\binom{N- w}{k_i},  \quad (\mbox{mod }2 ).
\end{equation}
\item[(2)] $+1$ - eigenstate of $-X^{\otimes N}$ iff for $0\leq w\leq N$
\begin{equation}
\sum_{i=1}^m\binom{ w}{k_i}=\sum_{i=1}^m\binom{N- w}{k_i}+1,  \quad (\mbox{mod }2 ).
\end{equation}
\item[(3)] $+1$ - eigenstate of $Y^{\otimes N}$ iff for $0\leq w\leq N$
\begin{equation}
\sum_{i=1}^m\binom{ w}{k_i}=\sum_{i=1}^m\binom{N- w}{k_i}+ w+\frac{N}{2},  \quad (\mbox{mod }2 ).
\end{equation}
\end{itemize}
These cases are in fact comprehensive when $ \max_i k_i>2$, meaning that, unless we are dealing with graph states, symmetric hypergraph states can have only one of these three Pauli stabilizers.
\end{lemma}
We add another observation. For even $N$ we have that $\binom{ w}{2}=\binom{N- w}{2}+ w+N/2$.
With the palindrome conditions above, it is evident that by adding/removing all hyperedges with cardinality 2 on even-qubit hypergraph states we can map $X^{\otimes N}$-stabilized states to $Y^{\otimes N}$-stabilized states and vice versa.
 
\section{The geometric measure of hypergraph states} 
The \textit{geometric measure of entanglement} is an entanglement measure, denoted by $E_G$, and is defined to be one minus the maximal squared overlap between a given state and the closest product state: 

\begin{equation}\label{eq:geometricMeasure}
    E_G(\ket{\psi})=1-\max_{\ket{\phi}=\ket{a}\ket{b}\ket{c}\dots} |\braket{\phi}{\psi}|^2.
\end{equation}
For a pure state it quantifies the distance to the set of separable states~\cite{wei2003geometric,Schimony1995Degree}. In general such optimization problems are difficult to handle analytically, since one needs to optimize over increasing number of parameters in the  multipartite case. On the other hand, there are several tricks that one can use to make calculations easier: (i) If the state can be mapped under local unitaries to another state with nonnegative real coefficients, then the optimization over the new state can be done using real product states. This can be seen by observing that if $a_i\geq 0$ for all $i$, the triangle inequality $|\sum_i a_i b_i|\leq \sum_i |a_i||b_i|=\sum_i a_i|b_i|$ is an equality if we also have $b_i\geq 0$ for all $i$. Note that local unitary operations do not change entanglement properties of a state they are applied to. Due to normalization constraints, this reduces the number of optimization parameters by half. (ii) Moreover, if the given multipartite state is permutation symmetric and has more than two parties, the closest product state also turns out to be permutation  symmetric~\cite{Robert2009Geometric}. To sum up, if we have permutation symmetric states with nonnegative coefficients, the problem can be reduced to a one-parameter optimization. In this section we show that for certain classes of hypergraph states we can use the tricks above to evaluate their entanglement.  

Let us consider an example of the four-qubit 3-uniform complete hypergraph state $\ket{H_4^3}$ given in Fig.~\ref{fig:figure1} (a). One can directly check, that the unitary matrix $U^{\otimes 4}$ with 
\begin{equation}\label{app:UnitaryH43}
U=\left(
\begin{array}{cc}
 \cos (t) & \sin (t) \\
 -\sin (t) & \cos (t) \\
\end{array}
\right)
\end{equation} 
and parameter $t=1/2\arctan[1/2(\sqrt{5}-1)]$ transforms the hypergraph state $\ket{H^3_4}$ to:  
\begin{align}
\begin{split}
\ket{S_{H_4}}= & \frac{1}{8} (3-\sqrt{5})\big(\ket{0000}+\ket{1111}\big)+\frac{1}{8}(1+\sqrt{5})(\ket{0011}+\mbox{perm.}).
\end{split}
\end{align}
Hence, the problem reduces to a one-parameter optimization, which can be solved analytically. We shall omit a proof at this point, but it is straightforward to calculate that $E_G(\ket{H_4^3})=\frac{25-3 \sqrt{5}}{32}\approx 0.571619 $. This value was previously derived in Ref.~\cite{Guehne2014Hypergraph}, but only numerically.

The technique of mapping a state to all positive coefficients can be used for wider classes of hypergraph states. We can consider the following lemma as one of the examples.
\begin{lemma}\label{lemma:Hadamardmap}
Let $\ket{H_N}$ be an $N$-qubit hypergraph state corresponding to a hypergraph with one vertex which is contained in all hyperedges.
Then the geometric measure of entanglement of $\ket{H_N}$ can be calculated over real product vectors.
\end{lemma}

\begin{figure}[t]
\centering
  \includegraphics[scale=0.62]{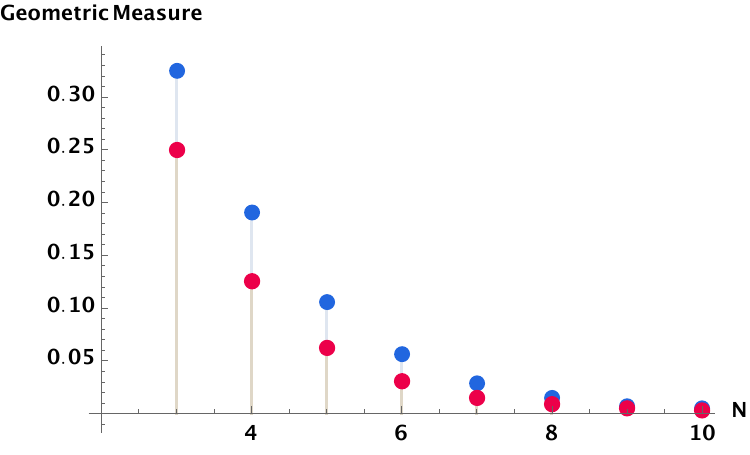}
  \caption{Entanglement of the $N$-qubit hypergraph states with a single $N$-cardinality hyperedge. The blue dots give the value of geometric measure of entanglement for $\ket{H_N^N}$ for $3$ to $10$ qubits. One can analytically check that the overlap maximization function has one global maximum. However, the expressions involved for the values are too cumbersome, thus, here we give their corresponding numerical values. As expected the entanglement quickly goes to zero as the number of qubits grow. The red dots (the lower ones) correspond to the lower bounds obtained in Ref.~\cite{Ghio2017Multipartite}, where the optimization is considered over bipartite states, instead of product ones.}
  \label{fig:figure2}
\end{figure}

\begin{proof}
After relabelling, we may assume that the vertex featured in all hyperedges corresponds to the first site. Then we can rewrite the state as
\begin{equation}
    \ket{H}=\frac{1}{\sqrt{2}}(\ket{0}\ket{+}^{\otimes N-1}+\ket{1}\ket{\tilde{H}}),
\end{equation}
where $\ket{\tilde{H}}$ is some hypergraph state over the remaining vertices. We now apply the Hadamard operator $\mathcal{H}$ to the first site:
\begin{equation}  
     \mathcal{H}_1\ket{H}=\frac{1}{\sqrt{2}}(\ket{+}^{\otimes V}+\ket{-}\ket{\tilde{H}}).
\end{equation}
Every computational basis entry appears with coefficient $\frac{+1}{\sqrt{2}^{N+1}}$ in the first summand and with $\frac{\pm 1}{\sqrt{2}^{N+1}}$ in the second and the contributions either add up to zero or $\frac{1}{\sqrt{2}^{N-1}}$. Thus, the resulting state has a real optimal product state overlap and since the Hadamard gate is also real, so does the original state $\ket{H}$.
\end{proof}

As an application we can directly calculate the geometric measure of entanglement of $N$-qubit hypergraph states $\ket{H_N^N}$ with a single hyperedge connecting all $N$ vertices (see Fig.~\ref{fig:figure1} (c) for an example). The overlap of $\ket{H_N^N}$ with a symmetric state $(a\ket{0}+b\ket{1})^{\otimes N}$ with $a,b\in\C$ is $\frac{1}{\sqrt{2}^N}\left((a+b)^N-2^N\right)$ and therefore
\begin{equation}\label{eq:NNoptimization}
    E_G(\ket{H^N_N})=1-\frac{1}{2^N}\max_{|a|^2+|b|^2=1}((a+b)^N-2b^N)^2.
\end{equation}
We can rewrite the hypergraph state by conditioning on the first qubit:
\begin{equation}
    \ket{H_N^N}=\frac{1}{\sqrt{2}}\left(\ket{0}\ket{+}^{\otimes N-1}+\ket{1}\ket{H_{N-1}^{N-1}}\right).
\end{equation}
Now, according to the previous observation, we can reduce the optimization in Eq.~\eqref{eq:NNoptimization} to run only over real numbers $a,b$ with $a^2+b^2=1$.
In Fig.~\ref{fig:figure2} the values of geometric measure of entanglement are plotted for different $N$. As expected, the state becomes very close to the product state when $N$ increases. Our results are also compared to the previously known lower-bounds, calculated in Ref.~\cite{Ghio2017Multipartite}. 

 Lemma~\ref{lemma:Hadamardmap} can also be used in non-symmetric cases to halve the number of parameters to optimize over. To give an example, consider the hypergraph state $\ket{H_4}_b$ in Fig.~\ref{fig:figure1} (b). If one applies Hadamard gates on qubits $1$ and $2$ one obtains the following state:
\begin{equation}
\ket{\tilde H_4}_b=\frac{1}{2}
(\ket{0000}+\ket{0001}+\ket{0010}+\ket{1111}),
\end{equation}
 for which the analytic value of the geometric measure of entanglement, $E_G(\ket{H_4}_b)=(5-\sqrt{5})/8$ can be directly calculated using derivatives of the two-variable function. 

We now consider a more systematic approach to analytically calculate the geometric measure of entanglement for several classes of symmetric hypergraph states. Specifically, we subsequently focus our attention to hypergraph states which exhibit local Pauli-symmetries. It turns out that for low hyperedge cardinalities, we can use local square roots of these Pauli stabilizers to map hypergraph states to real positive coefficient vectors. A well-known example of this is the GHZ state.  The fully-connected $2$-uniform hypergraph states $\ket{H_N^2}$ which are stabilized by $X^{\otimes N} $ can be mapped to the $N$-qubit GHZ state up to a global phase using square roots of the local Pauli-$X$ stabilizer: 
\begin{equation}
    \sqrt{X}_+^{\otimes N}\ket{H_N^2} =\pm \frac{1}{\sqrt{2}}(\ket{0\dots 0}+\ket{1\dots 1}).
\end{equation}
In all the other cases, a relative phase occurs, which can be corrected with the additional application of local Clifford phase gates.
We generalize this result to fully-connected $3$-uniform complete hypergraph states $\ket{H_N^3}$ and write them in a very convenient way. 
\begin{lemma}\label{thm:theorem1}
The even-qubit $3$-uniform complete hypergraph states can be mapped to a superposition of the GHZ state and all odd weight vectors after applications of square roots of local Pauli operators, corresponding to the respective stabilizers.
\vspace{5pt}\\
$\mathbf{N\equiv 2\mod 4}$: $\ket{H^3_N}$ is stabilized by $X^{\otimes N}$ and with $\ket{\tilde{H}^3_N}=\sqrt{X}_+^{\otimes N}\ket{H^3_N}$ we have 
\begin{equation}\label{eq:XthreeUnif}
\ket{\tilde{H}_N^3}=\pm\frac{1}{\sqrt{2}}\ket{GHZ}+\frac{1}{\sqrt{2}^N}\sum_{w(x)\, \mathrm{odd}}\ket{x}.
\end{equation}
$\mathbf{N\equiv 0\mod 4}$: $\ket{H^3_N}$ is stabilized by $Y^{\otimes N}$ and with $\ket{\tilde{H}^3_N}=\sqrt{Y}_+^{\otimes N}\ket{H^3_N}$, we have
\begin{equation}\label{eq:YthreeUnif}
\ket{\tilde{H}_N^3}=\pm\frac{1}{\sqrt{2}}\ket{GHZ}+\frac{1}{\sqrt{2}^N}\sum_{w(x)\, \mathrm{odd}}(-i)^{w(x)-1}\ket{x}.
\end{equation}
\end{lemma}
The proof of the lemma is given in \ref{proof:3-uniform}. Additionally, in \ref{sec: Appendix Geomeasure}, we derive most generally the state vector after application of the square roots of local Pauli-$X$ and $Y$ operators, see Eq.~\eqref{eq:XCoefficientsGeomeasure-2} and Eq.~\eqref{eq:YCoefficientsGeomeasure} for the closed formulae. In the $Y$-stabilized case, we can get rid of the alternating signs on the odd weights by applying an extra $\sqrt{Z}_+$ on each site, however we now pick up an imaginary part in the process.
The negative sign in front of the GHZ state occurs when the number of qubits is $N=8k+4$ or $N=8k+6$, respectively. By applying local Pauli-$Z$ to every site and neglecting the global sign, we can always achieve without loss of generality that the sign in front of the GHZ part is positive.
For $(3,2)$-uniform hypergraph states we can derive very similar results, see Lemma~\ref{lemma:3,2-uniform} in \ref{sec_coefficients} for the exact statement.

From the new way of writing the $3$-uniform complete hypergraph state we can derive an analytical expression for the geometric measure of entanglement for $\ket{H_N^3}$.  

\begin{proposition}\label{lemma:3-unif_GEM_X}
The geometric measure of entanglement of an $N$-qubit $3$-uniform complete hypergraph state with $N\equiv 2\mod 4$, i.e. it is stabilized by $X^{\otimes N}$, can be obtained using the following expression:
\begin{equation}\label{eq:GMHN3}
    E_G(\ket{H_N^3})=\frac{3}{4}-\frac{1}{\sqrt{2}^N}-\frac{1}{2^N}.
\end{equation}
\end{proposition}

\begin{proof}

Using the identity in Eq.~\eqref{eq:XthreeUnif}, we can express the overlap of $\ket{\tilde{H}^3_N}$ with $\ket{\psi(\theta)}^{\otimes N}=(\cos(\theta)\ket{0}+\sin(\theta)\ket{1})^{\otimes N}$ by
\begin{align*}
    f_N(\theta):=&\braket{\tilde{H}^3_N}{\psi(\theta)}^{\otimes N}\\
    = &\frac{1}{2}\left(\cos^N(\theta)+\sin^N(\theta)+\Big(\frac{1}{\sqrt{2}}(\cos(\theta)+\sin(\theta))\Big)^N-\Big(\frac{1}{\sqrt{2}}(\cos(\theta)-\sin(\theta))\Big)^N\right)\\
    =& \frac{1}{2}\left(\cos^N(\theta)+\sin^N(\theta)+\sin^N\left(\theta+\frac{\pi}{4}\right)-\sin^N\left(\theta-\frac{\pi}{4}\right)\right).
 \end{align*}
In order to determine the maximum of $f_N(\theta)$, we rewrite 

\[f_N(\theta)=\frac{1}{2}\left(\sum_{j=0}^3\cos^N(\theta+\frac{\pi j}{4})\right)-\cos^N(\theta+\frac{\pi}{4}).\] 

As shown in Appendix A, Lemma \ref{lemma:cyclic_max}, the sum $\sum_{j=0}^3\cos^N(\theta+\frac{\pi j}{4})$ becomes maximal at $\theta\in\frac{\pi}{4}\Z$. We observe, that for $\theta=\frac{\pi}{4}$ the negative contribution $\cos(\theta+\frac{\pi}{4})$ is zero, and since $N$ is even, therefore also minimal. We can then conclude that \[\max_\theta f_N(\theta)=f_N\left(\frac{\pi}{4}\right)=\frac{1}{2}+\frac{1}{\sqrt{2}^N}.\] Inserting this value, the geometric measure can be exactly calculated using the equality Eq.~(\ref{eq:GMHN3}). Thus, the geometric measure of entanglement for the $3$-uniform complete hypergraph state quickly approaches $3/4$.
\end{proof}
Next we treat the Pauli-Y stabilized case.

\begin{proposition}\label{lemma:3-unif_GEM_Y}
The geometric measure of an $N$-qubit $3$-uniform hypergraph state with $N\equiv 0\mod 4$, i.e. it is stabilized by $Y^{\otimes N}$, can be estimated as follows:
\begin{equation}\label{eq:Y-GEM}
\frac{3}{4}-\frac{1}{2^N}-\frac{1}{\sqrt{2}^N}\leq E_G\left(\ket{H^3_N}\right)\leq\frac{3}{4}-\frac{1}{2^N}.
\end{equation}
Numerical computations for different $N$ indicate, that in both estimates equality does not hold in general, see Fig.~\ref{fig:GEM}.
\end{proposition}

\begin{proof}
Again, the starting point is the identity in Eq.~\eqref{eq:YthreeUnif}. 
Unfortunately, the considered state does not have all positive coefficients anymore, so the closest product state need not have real coefficients. 
However, exploiting permutation symmetry, we can assume it to be of the form $\ket{\psi_N(\theta,\varphi)}:=\left(\cos(\theta)\ket{0}+e^{i\varphi}\sin(\theta)\ket{1} 
\right)^{\otimes N}$. The overlap is then computed similarly as before, one just needs to keep track of the occurring phases. It is then straightforward to derive the estimate
\begin{equation}
 \lvert\braket{\tilde{H}^3_N}{\psi_N(\theta,\varphi)}\rvert\leq|f_N(\theta)|\leq \frac{1}{2}+\frac{1}{\sqrt{2}^N}
 \end{equation} 
for all $\phi,\theta,N$, which proves to the lower bound in \eqref{eq:Y-GEM}. The upper bound on the geometric measure is obtained by inserting $\varphi=0,\theta=\pi/4$.
\end{proof}

Next, we extend our results to $5$-uniform complete hypergraph states, and more generally to ones with hyperedge cardinality of the type $(2^r+1)$ and which are 
stabilized either by $X^{\otimes{N}}$ or $Y^{\otimes N}$ operators with the $"+1"$ eigenvalue. In these cases, the sequence of binomial coefficients $\binom{w}{2^r+1}$ has a nice structure, which allows us to derive local unitary equivalences in the spirit of Lemma \ref{thm:theorem1}. We continue by using these equivalences to obtain analytical results on geometric measure of entanglement in the $5$-uniform case, and conjecture some lower bounds in the general case.

It turns out that the second part of Lemma~\ref{thm:theorem1} is a special case of the following result, which is proven in \ref{proof:lemma_Y_coeff_general}.
\begin{lemma}\label{thm:Ygeneral}
Let $\ket{H^{\mathbf{k}}_N}=\frac{1}{\sqrt{2}^N}\sum_{x}(-1)^{f(w(x))}\ket{x}$ be a symmetric $\mathbf{k}$-uniform $N$-qubit hypergraph state, which is stabilized by $Y^{\otimes N}$. Further, assume that $f$ is $2^r$-periodic, $f(w)=0$ for all even $w$ and $N\equiv 0\mod 2^{r-1}$. Then the state can be mapped to a superposition of the GHZ state and some odd weights by applying $\sqrt{Y}_+^{\otimes N}$, i.e.
\begin{equation}\label{eq:Ygeneral}
\ket{\widetilde{H}^\mathbf{k}_N}=\sqrt{Y}_+^{\otimes N}\ket{H^{\mathbf{k}}_N}=\frac{1}{\sqrt{2}}\ket{GHZ}+\frac{1}{\sqrt{2}}\ket{\phi_{\mathrm{odd}}^\mathbf{k}},
\end{equation} 
where $\ket{\phi_{\mathrm{odd}}^\mathbf{k}}$ is a normalized quantum state depending on $\mathbf{k}$ which features only odd weight contributions in the computational basis.
\end{lemma}
The conditions look somewhat artificial at the first sight. However, additionally to all single-cardinality hyperedge $\ket{H^k_N}$ which are stabilized by $Y^{\otimes N}$, they encompass e.g. $24$-qubit $(5,9)$-uniform or $16$-qubit $(3,5,9)$-uniform hypergraph states. 

\begin{lemma}\label{lemma:Xcoeff_general}
Let $r\geq 3$ and $N=l2^r+2^{r-1}$. Then the $k=(2^{r-1}+1)$-uniform complete $N$-qubit hypergraph state is stabilized by $X^{\otimes N}$ and we have
\begin{equation}
\ket{\tilde{H}^k_N}=\sqrt{X}_+^{\otimes N}\ket{H^k_N}=\frac{1}{\sqrt{2}}\ket{GHZ_X}+(-1)^l\sum_{w(x)\, \mathrm{odd}} c_{w(x)}\ket{x},
\end{equation}
where $\ket{GHZ_X}=\frac{1}{\sqrt{2}}(\ket{+}^{\otimes N}+\ket{-}^{\otimes N})$ and the odd coefficients are given by
\begin{equation}\label{eq:Xcoeff_general_odd}
c_w=\frac{2}{2^{r}}\sum_{j=1,\, \mathrm{odd}}^{2^r-1}i^{j-1}\frac{\cos^{N-w}\left(\frac{\pi j}{2^r}\right)\sin^{w}\left(\frac{\pi j}{2^r}\right)}{\cos\left(\frac{2\pi j}{2^{r}}\right)}.
\end{equation}
\end{lemma}

\begin{lemma}\label{lemma:Ycoeff_general}
Let $r\geq 3$ and $N=l2^r$. Then the $k=(2^{r-1}+1)$-uniform complete $N$-qubit hypergraph state is stabilized by $Y^{\otimes N}$ and we have 
\begin{equation}
\ket{\tilde{H}^k_N}=\sqrt{Y}_+^{\otimes N}\ket{H^k_N}=\frac{1}{\sqrt{2}}\ket{GHZ}+\sum_{w(x)\, \mathrm{odd}}i^{w-1}c_{w(x)}\ket{x},
\end{equation}
where the odd coefficients are given by
\begin{equation}
c_w=\frac{2}{2^{r}}\sum_{j=1,\, \mathrm{odd}}^{2^r-1}\frac{\cos^{N-w}(\frac{\pi j}{2^r})\sin^{w}(\frac{\pi j}{2^r})}{\sin(\frac{2\pi j}{2^{r}})}.
\end{equation}
\end{lemma}
Proofs of the lemmata are given in \ref{proofs:general-uniform}. Next, we look at the geometric  measure of $5$-uniform complete hypergraph states with local Pauli stabilizers. In the $5$-uniform case, the situation is still 
relatively well-behaved, since one can easily compute the coefficients in Eq.~\eqref{eq:Xcoeff_general_odd} to be

\begin{equation}
c_w=\frac{(-1)^l}{\sqrt{2}}\left(\cos^{N-w}\left(\frac{\pi}{8}\right)\sin^{w}\left(\frac{\pi}{8}\right)+\sin^{N-w}\left(\frac{\pi}{8}\right)\cos^w\left(\frac{\pi}{8}\right)\right).
\end{equation}

In particular (by applying local Pauli-Z if necessary), we have a local unitary transformation which maps $X$-stabilized $5$-uniform hypergraph states to states with nonnegative
coefficients. The closest product state to $\ket{H^5_N}$ (when $N\equiv 4\mod 8$) therefore has to be be assumed to have the form $(\sin\theta\ket{0}+\cos\theta\ket{1})^{\otimes N}$ for some $\theta\in[0,
\pi/2]$. After a few algebraic transformations on the analytic expression for the overlap, we can use the proof strategy for the optimization in the $3$-uniform setting. Again, we provide an exact value 
for the $X^{\otimes N}$-stabilized states and estimates for the $Y^{\otimes N}$-stabilized cases, both converging to $3/4$ as $N$ increases.
\begin{theorem}\label{thm:5-unif_GEM_X}  
Let $\ket{H^5_N}$ be a hypergraph state which is either stabilized by $X^{\otimes N}$, or $Y^{\otimes N}$. Then we have the following analytical results for the geometric entanglement measure of $\ket{H^5_N}$:
\begin{enumerate}
\item If $\ket{H^5_N}$ is $X^{\otimes N}$-stable, we have the exact expression:
\begin{equation}\label{eq:5_unif_GEM_X}
E_G(\ket{H^5_N})=\frac{3}{4}-\lambda_N-\lambda_N^2,
\end{equation}
where 

$\lambda_N=\frac{1}{\sqrt{2}}\left(\cos^N\left(\frac{\pi}{8}\right)-\sin^N\left(\frac{\pi}{8}\right)\right).$

\item If $\ket{H^5_N}$ is $Y^{\otimes N}$-stabilized, we can estimate
\begin{equation}\label{eq:5_unif_GEM_Y}
\frac{3}{4}-\lambda_N-\lambda_N^2\leq E_G(\ket{H^5_N})\leq\frac{3}{4}.
\end{equation}
\end{enumerate}
\end{theorem}
The proof can be found in~\ref{proof:5-uniform_GEM}. By exploiting the symmetries of the considered states at hand, we have improved significantly on the more universal bounds derived in Ref.~\cite{Ghio2017Multipartite}. Numerical experiments indicate that the qualitative results of Propositions \ref{thm:5-unif_GEM_X}, \ref{lemma:3-unif_GEM_Y}  and Theorem \ref{lemma:3-unif_GEM_X}, about the behaviour for large $N$ extend to complete $(2^{r}+1)$-uniform hypergraph states with local Pauli stabilizers. However, for $r\geq 3$, the involved terms behave not as nicely anymore, which prevents us from using the same strategy of proof as for $3$-, and $5$-uniform states. Moreover, the coefficients in Lemmata~\ref{lemma:Xcoeff_general} and \ref{lemma:Ycoeff_general} in general can also be negative, therefore the best one can hope for 
are estimates, similar to the ones obtained in the $3$- and $5$-uniform $Y^{\otimes N}$-stabilized cases above.\\
We conjecture that the geometric entanglement measure of $(2^{r-1}+1)$-uniform complete hypergraph states which are stabilized by local Pauli-$X$ satisfies the lower bound
\begin{equation}\label{eq:conj}
\frac{3}{4}-\lambda_N-\lambda_N^2\leq E_G(\ket{H^{2^{r-1}+1}_N})
\end{equation}
with
\begin{equation}
\lambda_N:=\frac{2}{2^r}\sum_{j=1,\, \mathrm{odd}}^{2^{r}-1}\sgn\left(\cos(\frac{\pi j}{2^r})\right)\frac{\cos^N(\frac{\pi}{4}-\frac{\pi j}{2^r})}{|\cos(\frac{2\pi j}{2^r})|}.
\end{equation}
On the other hand, the overlap of $\ket{\tilde{H}^k_N}$ with one of the two states $(\frac{\ket{0}\pm\ket{1}}{\sqrt{2}})^{\otimes N}$ must be greater than $1/2$. Therefore, when the hyperedge cardinality is fixed, the geometric entanglement measure converges to $3/4$ as $N$ grows. Considering the previous examples, which hint at $Y^{\otimes N}$-stabilized states intrinsically having a higher geometric entanglement measure than $X^{\otimes N}$-stabilized ones, it seems natural to expect that the estimate Eq.\eqref{eq:conj} also carries over to $Y^{\otimes N}$-stabilized states. Here, we obtain the general upper bound $3/4$ by observing that the overlap of $\ket{\tilde{H}^k_N}$ with $\ket{0}^{\otimes N}$ is always $1/2$.
Much weaker lower estimates on the geometric measure of hypergraph states, based on entanglement witnessing using the maximal Schmidt coefficient in any bipartition, are also discussed in Ref.~\cite{Ghio2017Multipartite}. 
 The results of several numerical computations for the geometric entanglement measure of different hypergraph states are plotted in Fig.~\ref{fig:GEM}. It shows that the lower bounds on the entanglement measure for $Y^{\otimes N}$-stabilized states, which we proved and conjectured are in general not saturated.
\begin{figure}\label{fig:GEM}
\centering
\includegraphics[scale=1]{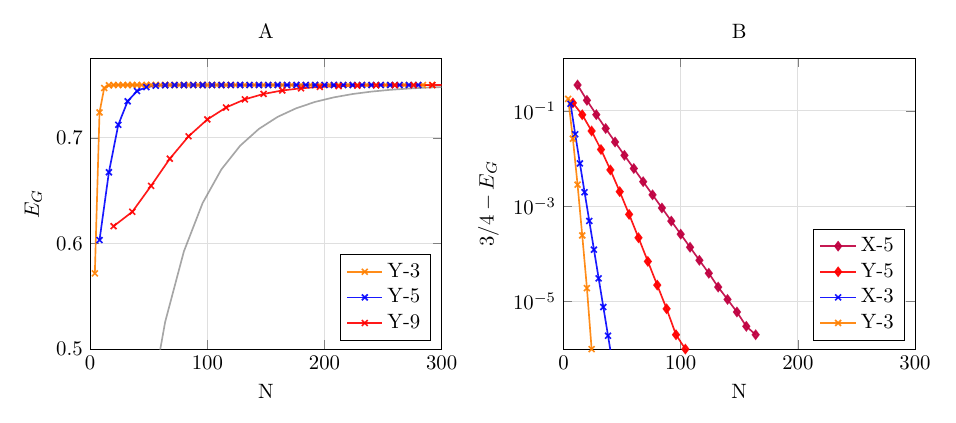}\caption{\textbf{A}: Numerically obtained values for the geometric measure of entanglement for different classes of hypergraph states, which feature $Y^{\otimes N}$-symmetries. The additional grey curve represents the lower bound for the $9$-uniform case conjectured in Eq.~\eqref{eq:conj}. \textbf{B}: Comparison of the difference of geometric measure of entanglement for different $X^{\otimes N}$ and $Y^{\otimes N}$-stabilized states to the maximum value of $3/4$ on a logarithmic scale. Note, that the interpolated values for the $X^{\otimes N}$-stabilized states coincide with the lower bounds in Eq.~\eqref{eq:Y-GEM} and Eq.~\eqref{eq:5_unif_GEM_Y}, respectively.}
\end{figure}

\section{Connection to nonlocality}
\subsection{Exponential violation of local realism}
It is known that complete hypergraph states violate Mermin-like inequalities in an exponential and robust manner~\cite{gachechiladze2016Extreme}. Both, the violation and the robustness results were highly technical, since they were derived using brute-force calculations and lacked physical intuition. Here we reformulate the results 
in a systematic way and then expand them further to infinitely many classes of hypergraph states. The key idea in all of the subsequent considerations is the following: In order to calculate the quantum value $\mean{\mathcal{B}_N}_{\ket{H_N}}$ of the considered Mermin-inequality $\mathcal{B}_N$ as in Eq.~\eqref{eq:MerminOperatorOur}, we augment the expression $\bra{H_N}\mathcal{B}_N\ket{H_N}=\bra{H_N}U^\dagger U\mathcal{B}_N U U^\dagger \ket{H_N}$. With a clever choice of the unitary $U$, we can get rid of the alternating signs in the Bell operator by considering $U\mathcal{B}_N U$ and also transform $\ket{H_N}$ to a state $U^\dagger\ket{H_N}$, which has a structure easier to handle computationally as opposed to the original one.

In  Ref.~\cite{gachechiladze2016Extreme} the violation was derived using Mermin-like operators of the form:
\begin{align}\label{eq:MerminOperatorOur}
\begin{split}
    \mathcal{B}^{P}_N  
    &=\frac{1}{2}\left((P+iZ)^{\otimes N}+(P-iZ)^{\otimes N}\right) \\
    &=\sum_{m=0,\,\mathrm{even}}^N i^m Z_1\dots Z_m P_{m+1}\dots P_N+\mathrm{ perm.}
\end{split}
\end{align}
for $P\in\{X,Y\}$.
Note that the expression is permutation symmetric. It was shown in Ref.~\cite{gachechiladze2016Extreme} that the quantum value of $\mathcal
{B}^X_N$ on $\ket{H^3_N}$ is $\mean{\mathcal{B}^X_N}_{\ket{H_N^3}}=2^{N-2}$, while 
within local realistic theories, it is bounded by  $2^{N/2}$~\cite{Mermin1990Extreme}. Here we derive the same result in a much more 
concise manner, exploiting local Pauli symmetries of the hypergraph state. The following discussion relies significantly on the following observation: When applying $\sqrt{P}_-^{\otimes N}$ or $\sqrt{Z}_-^{\otimes N}$ to both sides of $\mathcal{B}^P_N$, we can simplify the structure of $\mathcal{B}_N^P$:
\begin{align}\label{eq:tildeBell}
    \begin{split}
    \mathcal{\tilde{B}}_N&:=\sqrt{P}_-^{\otimes N}\mathcal{B}_N^P\sqrt{P}_-^{\otimes N}\\  
    &=\frac{1}{2}\left((\id+Z)^{\otimes N}+(\id-Z)^{\otimes N}\right) \\
    &=\sum_{m=0,\,\mathrm{even}}^N \id_1\dots \id_m Z_{m+1}\dots Z_N+\mathrm{ perm.}
\end{split}
\end{align}
where we used that $\sqrt{P}_-Z\sqrt{P}_-=\pm i Z$.
Similarly, we can also exploit that $\sqrt{Z}_-P\sqrt{Z}_-=\pm iP$ to find:
\begin{align}\label{eq:applyZonbothsides}
    \begin{split}
    \mathcal{\tilde{B}}^P_N&:=\sqrt{Z}_-^{\otimes N}\mathcal{B}_N^P\sqrt{Z}_-^{\otimes N}\\  
    &=\frac{1}{2}\left((\id+P)^{\otimes N}+(\id-P)^{\otimes N}\right) \\
    &=\sum_{m=0,\,\mathrm{even}}^N \id_1\dots \id_m P_{m+1}\dots P_N+\mathrm{ perm.}
\end{split}
\end{align}
As a working example, we consider the case where $N\equiv 2\mod 4$, i.e. $\ket{H^3_N}$ is invariant under $X^{\otimes N}$. We may 
rewrite $\mean{\mathcal{B}^X_N}_{\ket{H^3_N}}$ as
\begin{align}
\bra{H_N^3} \mathcal{B}^{X}_N \ket{H_N^3} &=  \bra{H_N^3}\sqrt{X}_{+}^{\otimes N}\tilde{\mathcal{B}}_N \sqrt{X}_{+}^{\otimes N}\ket{H_N^3}= \bra{\tilde H_N^3} \tilde{\mathcal{B}}_N  \ket{\tilde H_N^3},
\end{align}
with the transformed operator $\tilde{\mathcal{B}}_N$ as in Eq.~\eqref{eq:tildeBell}.
According to Eq.~\eqref{eq:XthreeUnif}, we also have $\ket{\tilde H_N^3}=\frac{\pm 1}{\sqrt{2}}\ket{GHZ}+\frac{1}{\sqrt{2}}\ket{\phi_{\mathrm{odd}}}$, where $\ket{\phi_{\mathrm{odd}}}=\frac{1}{\sqrt{2}^{N-1}}\sum_{w(x)\,\mathrm{odd}}\ket{x}$ contains all the odd weight contributions.
Due to linearity of the expression $\bra{\tilde{H}_N^3} \tilde{\mathcal{B}}_N \ket{\tilde{H}_N^3}$, we can separate it into four parts. The summands of $
\tilde{\mathcal{B}}_N$ can only flip a sign when acting on the computational basis. We immediately conclude that the cross-terms do not yield any contribution, $\bra{GHZ} \tilde{\mathcal{B}}_N\ket{\phi_{\mathrm{odd}}}=0$, and 
therefore:
\begin{equation}
\mean{\mathcal{B}^X_N}_{\ket{H^3_N}}=\frac{1}{2}\bra{GHZ}\tilde{\mathcal{B}}_N\ket{GHZ}+\frac{1}{2}\bra{\phi_{\mathrm{odd}}}\tilde{\mathcal{B}}_N\ket{\phi_{\mathrm{odd}}}.
 \end{equation}
Let us consider the GHZ part first. Since all the terms in \eqref{eq:tildeBell} appear with a positive sign and the number of Pauli-$Z$ gates is always even, $\ket{GHZ}$ is a $+1$-eigenstate of every summand of $\tilde{\mathcal{B}}_N$ and we directly obtain: \begin{equation}\label{eq:GHZ-contribution}
        \frac{1}{2} \bra{GHZ} \tilde{\mathcal{B}}_N \ket{GHZ}=\frac{1}{2} 2^{N-1} =2^{N-2}.
    \end{equation}
We consider the last case and the action of Pauli-$Z$ on $\ket{\phi_{\mathrm{odd}}}$. Because $N$ is even, we clearly have $\mathcal{\tilde{B}}_N =\tilde{\mathcal{B}}_N Z^{\otimes N}$. Further, $Z^{\otimes N}\ket{x}=-\ket{x}$ for every odd weight computational basis state $\ket{x}$. Hence,
\begin{equation}\label{eq:applyallzsforsign}
\bra{\phi_{\mathrm{odd}}}\tilde{\mathcal{B}}_N\ket{\phi_{\mathrm{odd}}}=\bra{\phi_{\mathrm{odd}}}\tilde{\mathcal{B}}_NZ^{\otimes N}\ket{\phi_{\mathrm{odd}}}=-\bra{\phi_{\mathrm{odd}}}\tilde{\mathcal{B}}_N\ket{\phi_{\mathrm{odd}}}
\end{equation}
which consequently must be zero.
This finishes the entire derivation proving that 
\begin{equation}
     \bra{H_N^3} \mathcal{B}^X_N \ket{H_N^3} =2^{N-2}.
\end{equation}
It is evident that the sign in front of the GHZ part in $\ket{\tilde{H}_N^3}$ does not make any difference.
This reasoning can directly be transferred to arbitrary hypergraph states which get mapped to a superposition of the GHZ state and odd weight vectors after applying local square roots of the respective stabilizer.
\begin{theorem}\label{thm:General_Bell_violation}
Let $\ket{H^{\mathbf{k}}_N}$ be a $\mathbf{k}$-uniform complete hypergraph state, which is stabilized by $P^{\otimes N}$, $P\in\{X,Y\}$ and gets mapped to a superposition of $\ket{GHZ}$ and odd weight vectors via:
\begin{equation}
\sqrt{P}_+^{\otimes N}\ket{H^{\mathbf{k}}_N}=\pm\frac{1}{\sqrt{2}}\ket{GHZ}+\frac{1}{\sqrt{2}}\ket{\phi_{\mathrm{odd}}^{\mathbf{k}}}.
\end{equation}
Then the quantum value of the Mermin-operator $\mathcal{B}_N^P$ evaluated on $\ket{H^{\mathbf{k}}_N}$ is
\begin{equation}
\bra{H^\mathbf{k}_N}\mathcal{B}^P_N\ket{H^\mathbf{k}_N}=2^{N-2},
\end{equation}
as opposed to the classical bound of $\sqrt{2}^N$ for any local realistic theory.
\end{theorem}
\begin{proof}
Since we have  $\sqrt{P}_+Z\sqrt{P}_+=\pm iZ$ for both $P\in\{X,Y\}$, we can conclude $\sqrt{P}_+^{\otimes N}\mathcal{B}^P_N\sqrt{P}_+^{\otimes N}=\tilde{\mathcal{B}}_N$ as in Eq.~\eqref{eq:tildeBell}. The argument now works exactly like in our previous example. The cross contributions $\bra{GHZ}\tilde{\mathcal{B}}_N\ket{\phi^{\mathbf{k}}}$ are both equal to zero, and the GHZ part gives $\bra{GHZ}\tilde{\mathcal{B}}_N\ket{GHZ}=2^{N-2}$ as in Eq.~\eqref{eq:GHZ-contribution}. The reasoning for Eq.~\eqref{eq:applyallzsforsign} also immediately carries over to the more general setting and we can conclude that again$\bra{\phi^{\mathbf{k}}}\tilde{\mathcal{B}}_N\ket{\phi^{\mathbf{k}}}=0$.
\end{proof}
As evident from Lemma~\ref{thm:theorem1} and the follow-up Lemma~\ref{lemma:3,2-uniform} (Appendix only), this theorem yields exponential violation of Mermin-type inequalities for all symmetric even-qubit $3$- and $(3,2)$-uniform hypergraph states. Several other classes of hypergraph states which satisfy the assumptions of Theorem~\ref{thm:General_Bell_violation} are listed in Tab.~\ref{tab:GHZ odd_hypergraph_states}. 

\begin{table}\centering
\begin{tabular}{c|c|r|c}
$\mathbf{k}$ & $ m $ & $ N$\qquad & Stabilizer\\
\hline
$(k_m,\dots,k_1,2)$& odd & $0\mod n_{r}$ & $X^{\otimes N}$\\
$(k_m,\dots,k_1,2)$ & even & $ \frac{n_r}{2}\mod n_{r} $ & $X^{\otimes N}$\\
$(k_m,\dots,k_1)$& odd & $0\mod n_{r}$ & $Y^{\otimes N}$\\
$(k_m,\dots,k_1)$ & even & $ \frac{n_r}{2}\mod n_{r} $ & $Y^{\otimes N}$\\
\end{tabular}\vspace{2pt}\\
\caption{\label{tab:GHZ odd_hypergraph_states} Classes of hypergraph states $\ket{H^{\mathbf{k}}_N}$, satisfying the assumption of Theorem~\ref{thm:General_Bell_violation}. Here, we choose any hypergraph with the cardinality vector $\mathbf{k}$, with the entries of the form $k_i=2^{i}+1$, $i\geq 1$. Choose the highest cardinality, $k_{\max}$  and compute $r=\log_2(k_{\max}-1)$. Set $n_r=2^{r+1}$ and pick $N$ according to the third column. Then the fourth column indicates the corresponding stabilizer.  Observe the interplay between $X$- and $Y$-stabilized states in the sense of the remark after Lemma \ref{lemma:localStab}. The $Y$-stabilized cases where $m$ is odd are covered by Lemma~\ref{thm:Ygeneral}. Using the coefficient formulae in Lemma \ref{lemma:Xcoeff_general}, it is also possible to give a rigorous mathematical proof for the other cases listed here. However, they is tedious to derive and the result is not of the crucial importance for this work, therefore we shall omit the proofs here.}
\end{table}

In order to achieve a violation for $(k=2^{r-1}+1)$-uniform $X^{\otimes N}$-stabilized states, it seems natural to consider the operator $\mathcal{B}^X_N$. However, this is in fact a dead end, since the even weight contributions of the transformed state only give $\bra{GHZ_X}\tilde{\mathcal{B}}_N\ket{GHZ_X}=1$, whilst the odd ones and the cross-terms still yield zero. Instead, we consider the operator $\mathcal{B}^Y_N$ and rewrite its quantum value on $\ket{H^k_N}$ by
\begin{equation}
\bra{H^k_N}\mathcal{B}^Y_N\ket{H^k_N}=\bra{\tilde{H}^k_N}\sqrt{X}_-^{\otimes N}\mathcal{B}^Y_N\sqrt{X}_-^{\otimes N}\ket{\tilde{H}^k_N},
\end{equation}
with $\ket{\tilde{H}^k_N}=\frac{1}{\sqrt{2}}\ket{GHZ_X}+\frac{1}{\sqrt{2}}\ket{\phi_{\mathrm{odd}}}.$ We continue with the following observations: First, the transformation actually does not change the Mermin operator: $\sqrt{X}_-^{\otimes N}\mathcal{B}^Y_N\sqrt{X}_-^{\otimes N}=\mathcal{B}^Y_N$, which is easily checked. Secondly, because $N\equiv 0\mod 4$, we can express $\mathcal{B}^Y_N$ via $\mathcal{B}^Y_N=Z^{\otimes N}\tilde{\mathcal{B}}^X_N$
with the transformed Mermin-type operator $\mathcal{\tilde{B}}_N^X$ as in Eq.~\eqref{eq:applyZonbothsides}.
Now $\ket{GHZ_X}$ is a $+1$ eigenstate both of $Z^{\otimes N}$ and each summand of $\mathcal{\tilde{B}}^X_N$, due to the even number of Pauli-$X$ operators appearing in these and $N$ being even. Therefore we can conclude
\begin{equation}
\frac{1}{2}\bra{GHZ}\mathcal{B}^Y_N\ket{GHZ}=\frac{1}{2}\bra{GHZ_X}Z^{\otimes N}\mathcal{\tilde{B}}^X_N\ket{GHZ_X}=2^{N-2}.
\end{equation}
All summands of $\mathcal{B}^Y_N$ contain an even number of Pauli-$Y$'s, so they preserve weights modulo two and thus the cross-contributions vanish.
Therefore, the last contribution to consider are the odd weight contributions. Contrary to the previous two examples, this no longer vanishes in general. However, an explicit calculation (see \ref{sec:Xviolation}) using Lemma~\ref{lemma:Xcoeff_general}, shows that
\begin{equation}\label{eq:odd_part_contribution}
\frac{1}{2}\bra{\phi_{\mathrm{odd}}}\mathcal{B}^Y_N\ket{\phi_{\mathrm{odd}}}=-\frac{4}{4^r}\left|\sum_{l=1,\,\mathrm{odd}}^{2^r-1}i^l\frac{\left(\cos(\frac{l\pi}{2^r})+\sin(\frac{l\pi}{2^r})\right)^N}{\cos(\frac{2l\pi}{2^r})}\right|^2.
\end{equation}
The decisive point is now that the sum over $l$ only runs over odd integers, therefore we always have $\cos(\frac{l\pi}{2^r})+\sin(\frac{l\pi}{2^r})<2$. If we keep the hyperedge cardinality (and thereby $r$) fixed and increase $N$, the contribution in Eq.~\eqref{eq:odd_part_contribution} becomes more and more insignificant compared to the leading order $2^{N-2}$. As a result, we get an infinite family of hypergraph states achieving the exponential violation of local realism. Fig.~\ref{fig:Xvioation} shows this exponential violation for a few different classes of hypergraph states. 
\begin{figure}
    \centering
    \includegraphics{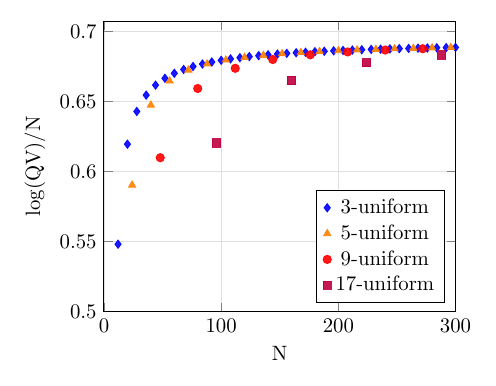}
    \caption{\label{fig:Xvioation} Quantum value $(QV)$ of $\mathcal{B}^Y_N$ evaluated on different hypergraph states stabilized by $X^{\otimes N}$, according to Eq.~\eqref{eq:odd_part_contribution}. The classical bound corresponds to the value $\log(\sqrt{2})\approx 0.347$}
    
\end{figure}

\subsection{Robustness against particle loss}
With a growing number of qubits, it becomes more and more idealistic to check to which extent nonlocality, or violation of Bell inequalities in hypergraph states, is stable under the loss of particles. It is known that the $4$-uniform complete hypergraph states keep the exponential violation even if several particles are lost~\cite{gachechiladze2019quantum}. However, for $3$-uniform complete hypergraph states, we do not obtain a Bell-violation using Mermin inequalities anymore after losing even a single particle. Instead, the resulting states exponentially violate the separability bounds of Mermin-like operators presented in Ref.~\cite{Roy2005Seperability} and thus are still entangled. For separable states, the quantum value of the Mermin-type operators $\mathcal{B}_N$ considered in the previous section, cannot exceed $\sqrt{2}$~\cite{Roy2005Seperability}.

Here we reproduce some results of Ref.~\cite{gachechiladze2019quantum}, making use of symmetries which allow for nicer representations of hypergraph states after loosing particles in the $3$-uniform case.
To give an example, assume that one or more ($k$) particles are lost during the computation. We then replace the corresponding entry in the Bell operator with the identity  operator, that is we calculate the separability violation for a reduced density matrix of the state. To obtain violations, we consider different Bell operators, depending on the initial state. However, their general structure is always of the form
\begin{equation}
     \mathcal{B}_{N\backslash k}:= \mathcal{M}^i_{N-k}\otimes \id^{\otimes k},\,i\in\{0,1\}
\end{equation}
where $\mathcal{M}^0$ and $\mathcal{M}^1$ are the Mermin-type operators given in the Appendix, Eq.~(\ref{eq:MerminOperatorOur0}) and Eq.~\eqref{eq:MerminOperatorOur1}, which feature only even/odd numbers of Pauli-$X$, respectively.
The notation $\mathcal{B}_{N\backslash k}$ indicates that $k$ particles are lost. Due to the symmetry, it is justified to fix them to be the ones corresponding to the last sites. In Ref.~\cite{gachechiladze2016Extreme} the robustness for $X^{\otimes N}$-stabilized $\ket{H^3_N}$ and $\mathcal{B}_{N\backslash k}=\mathcal{M}^0_{N-k}\otimes\id^{\otimes k}$ was derived only for $k=1$ and with a very lengthy proof. We simplify the calculations, to more general initial states and numbers of lost qubits.

\begin{theorem}\label{thm:robustness}
After loosing $k$ particles of $N$-qubit $3$-uniform complete hypergraph state, we can derive the following violations of separability inequalities:\vspace{8pt}\\
    \begin{tabular}{c|c|c}
    \textup{Constraints on} $N$ \textup{and} $k$ & \textup{Bell inequality} $\mathcal{B}_{N\backslash k}$ & \textup{Quantum value}
    $\mean{\mathcal{B}_{N\backslash k}}_{\ket{H^3_{N}}}$\\
    \hline
    $N\equiv 2\mod 4$, $k$ odd & $\mathcal{M}^0_{N-k}\otimes\id^{\otimes k}$ & $\sqrt{2}^{N-2k}$\\
    $N\equiv 0\mod 4$, $k$ even &$\mathcal{M}^1_{N-k}\otimes\id^{\otimes k}$ & $\sqrt{2}^{N-2k}$\\
    $N-k\equiv 2\mod 4$ & $\mathcal{M}^1_{N-k}\otimes\id^{\otimes k}$ & $\left|\sin\left(\frac{\pi k}{4}\right)\right|\sqrt{2}^{N-2k}$\\
    $N-k\equiv 0\mod 4$ & $\mathcal{M}^1_{N-k}\otimes\id^{\otimes k}$ & $\left|\cos\left(\frac{\pi k}{4}\right)\right|\sqrt{2}^{N-2k}$
\end{tabular}\vspace{5pt}\\
In the first two cases, the quantum value is $1/2$ or $0$ if we instead consider even and odd $k$, respectively.
\end{theorem}
The proof can be found in \ref{prf:thm_robustness}. As a consequence of this theorem and the fact that loosing a particle of a separable state cannot lead to an entangled state, we can deduce that $3$-uniform complete hypergraph states remain entangled if up to $\lfloor (N-4)/2\rfloor$ particles are lost. The first two cases guarantee this for even $N$, whereas the last two can be used when $N$ is odd. 
However, in order to detect non-locality rather than just entanglement after particle loss, Mermin-inequalities do not seem to be a very good choice. For $3$-uniform states it is not possible to detect nonlocality with our presented methods, however it is known that this is possible in the $4$-uniform case~\cite{gachechiladze2019quantum}.

\section{Conclusions and outlook}
We discuss how local Pauli symmetries of hypergraph states can aid in analysing entanglement properties and nonlocality of complete uniform hypergraph states. By transforming the state with local square-roots of Pauli operators, we analytically calculate the geometric measure of entanglement for various class of hypergraph states. Additionally, we significantly simplify calculations to recover previously derived results and extend them further to cover more cases of hypergraph states. Our results shed light and deeper understanding to the rich structures present in these states. Families of interesting hypergraph states correspond to the superposition of GHZ states with overwhelming amplitude and the Dicke states with exponentially decreasing amplitudes. This structure explains both, the exponential violation of Mermin inequality and robustness against particle losses. In this work, we derived robustness results only for seperability inequalities, and Mermin inequalities do not seem to be well suited for detecting non-locality in 3-uniform states after particle losses. Inequalities which might work better for this purpose are for instance WWWZB or Hardy-type inequalities. These have previously been investigated in the context of particle loss in Dicke states~\cite{Barnea2015Nonlocality,Divianszky2016Bounding,Wiesniak2021Symmetrized}. Numerical experiments of ours with WWWZB inequalities on 3-uniform states indicate that such results can also be derived for symmetric hypergraph states, e.g. for 8-qubit 3-uniform states if two particles are lost, for suitably chosen measurement bases, WWWZB inequality detects nonlocality of the reduced state. More general results, however, require further investigation. 

Local symmetries are key to all the findings in this work. Such symmetries have been exhaustively investigated for the complete uniform hypergraph states, restricted to the Pauli stabilizers only. On the other hand, recently, all local, invertible (unitary and nonunitary) symmetries of arbitrary stabilizer states  (graph states) have been investigated~\cite{englbrecht2020symmetries}, and connections to their entanglement structure and applications in quantum-error correction were identified. Motivated by this, exhaustive study of general symmetries for the states with nonlocal stabilizer could help better understand their entanglement properties. Besides, the identified symmetries could uncover additional transformations between hypergraph states and from hypergraph states to other multipartite pure states. Recall that extension of local complementation from graphs to hypergraphs was the central tool to construct a family of counterexamples to the famous LU=LC conjecture~\cite{ji2007lu,tsimakuridze2017graph}.

\section{Acknowledgments}
We thank Nikoloz Tsimakuridze, Géza Tóth, Nikolai Miklin, and David Wierichs for interesting ideas and discussions. Also, we would like to thank the anonymous referee for insightful comments concerning further research directions.
This work was supported by the Deutsche 
Forschungsgemeinschaft (DFG, German Research Foundation, project
numbers 236615297 - SFB 1119, 447948357 and 440958198), the Sino-German
Center for Research Promotion (Project M-0294), the
ERC (Consolidator Grant 683107/TempoQ), and the 
German Ministry of Education and Research (Project
QuKuK, BMBF Grant No. 16KIS1618K).

\appendix

\section{Appendix on Geometric measure}\label{sec: Appendix Geomeasure}
This part of the Appendix contains the proofs of the statements leading up to the results on the geometric measure of entanglement, which were not proven in the main text. 
In~\ref{sec_coefficients}, we derive coefficients after local square root transformations. In~\ref{proof:3-uniform}, we give the state transformation results for the $3$-uniform case. In~\ref{proofs:general-uniform}, we give computations on coefficients beyond the $3$-uniform case. In~\ref{sec:GeoMeasure}, we give analytical derivations for the geometric measure of entanglement. In~\ref{sec:nonlocality}, we discuss nonlocality results. In~\ref{sec:particleloss}, we give detailed calculations for the violation of the separability inequality in the $3$-uniform complete case in the presence of particle loss.

\subsection{Computing coefficients after local square root transformations - general results} \label{sec_coefficients}
First we state and prove some technical results, which simplify subsequent computations. 

\begin{lemma}\label{lemma:aux2}
Let $N,q\in\N$ and $0<n\leq N$. Then
\begin{align}
 \label{eq:modulo_sum}   \sum_{ w\equiv q \textup{ mod }n}\binom{N}{ w} = & \frac{1}{n}\bigg(\sum_{j=0}^{n-1}(2 \cos{\frac{\pi j}{n}})^{N} \cos\frac{\pi j(N-2q)}{n}\bigg).
\end{align}
In particular, in the form which we will encounter later, this reads
\begin{align}
\sum_{ w\equiv q \textup{ mod }2^r}\binom{N-|e|}{ w-m} = \frac{1}{2^r}\bigg(\sum_{j=0}^{2^r-1}(2 \cos{\frac{\pi j}{2^r}})^{N-|e|} \cos{\frac{\pi j(N-|e|+2m-2q)}{2^r}}\bigg).
\end{align}
\end{lemma}

\begin{proof}
Let $\zeta_n$ be a primitive $n$-th root of unity. Then the expression $\frac{1}{n}(1+(\zeta_n^1)^l+\dots+(\zeta_n^{n-1})^l)$ is non-zero only whenever $l\equiv 0\mod n$, in which case it equals $1$. We use this property of the roots of unity to rewrite the left hand side of Eq.~\eqref{eq:modulo_sum}:
\begin{align*}
\sum_{l=0}^{\lfloor(N-q)/l\rfloor}\binom{N}{nl+q}&=\sum_{l'=0}^N\binom{N}{l'}\frac{1}{n}\Re\left(1+(\zeta_n^1)^{l'-q}+\dots+(\zeta_n^{n-1})^{l'-q}\right)\\
&=\frac{1}{n}\sum_{j=0}^{n-1}\Re\left(e^{-2q\pi j i/n}\left(1+e^{2j\pi i/n}\right)^N\right)\\
&=\frac{1}{n}\sum_{j=0}^{n-1}\Re\left(e^{i\pi j(N-2q)/n}\left(e^{-i\pi j/n}+e^{i\pi j/n}\right)^N\right)\\
&=\frac{1}{n}\sum_{j=0}^{n-1}\left(2\cos\frac{\pi j}{n}\right)^N\cos\frac{\pi j(N-2q)}{n}.
\end{align*}
\end{proof}

\begin{lemma}\label{lemma:aux1}
For $\beta\neq 0$, $\alpha\in\N$ and $M\in\N$ we have
\begin{equation}\label{eq:cosIdentity}
\sum_{m=0}^M(-1)^m\binom{M}{m}\cos\frac{\pi(2m-M+\alpha)}{\beta}=\Re\left(e^{i\pi\alpha/\beta}(-2i\sin\frac{\pi}{\beta})^M\right).
\end{equation}
In particular, in the form it will appear later,
\begin{align}
\sum_{m=0}^{|e|}(-1)^m\binom{|e|}{m}\cos\frac{\pi j(2m+{N-|e|}-2q)}{2^r}=\Re\left(e^{i\pi j(N-2q)/2^r}(-2i\sin\frac{\pi j}{2^r})^{|e|}\right).
\end{align}
\end{lemma}

\begin{proof}
\begin{align}
&\sum_{m=0}^M(-1)^m\binom{M}{m}\cos\frac{\pi(2m-M+\alpha)}{\beta}=\sum_{m=0}^M(-1)^m\binom{M}{m}\Re\left(e^{i\frac{\pi(2m-M+\alpha)}{\beta}}\right)\\
&=\Re\left(e^{i\pi\alpha/\beta}e^{i\pi M/\beta}\sum_{m=0}^M\binom{M}{m}(-1)^m e^{2im\pi/\beta}\right)
=\Re\left(e^{i\pi\alpha/\beta}e^{i\pi M/\beta}(1-e^{2i\pi/\beta})^M\right)\\
&=\Re\left(e^{i\pi\alpha/\beta}(e^{-i\pi/\beta}-e^{i\pi/\beta})^M\right)=\Re\left(e^{i\pi\alpha/\beta}(-2i\sin\frac{\pi}{\beta})^M\right).
\end{align}
\end{proof}

\begin{proposition}\label{prop:Xformula}
\begin{enumerate}
\item Given a symmetric hypergraph state which is stabilized by $X^{\otimes N}$, then it is possible to calculate the coefficient $c_{|e|}=c_{1\dots10\dots 0}$ of the computational basis element of a weight $|e|$ after application of $\sqrt{X}^{\otimes N}$, using the following expression:
\begin{align}
\label{eq:XCoefficientsGeomeasure-1}c_{|e|}=&\frac{1}{(2\sqrt{2})^N}\sum_{ w}(-1)^{f( w)}\Re\left((1+i)^{N}(-i)^{ w+|e|}\right)\sum_{m=0}^{|e|}\binom{N-|e|}{ w-m}\binom{|e|}{m}(-1)^m\\
\label{eq:XCoefficientsGeomeasure-2}=&\frac{1}{\sqrt{2}^{N}}-\frac{2}{(2\sqrt{2})^N}\sum_{f( w)=1}\sum_{m=0}^{|e|}\binom{N-|e|}{ w-m}\binom{|e|}{m}\Re\left((1+i)^{N}(-i)^{ w+|e|-2m}\right).
\end{align}
Here $f:\{0,\dots,N\}\to\{0,1\}$ is the function which describes the symmetric hypergraph state by specifying the sign of different weights, encompassing the hyperedge cardinalities.
\item If the function $f$ is $2^r$-periodic, i.e. $f( w)=f( w+2^{r})$, Eq.~\eqref{eq:XCoefficientsGeomeasure-1} can alternatively be expressed as
\begin{align}
\label{eq:XCoefficientsGeomeasure2-1}c_{|e|}=&\frac{1}{2^r}\sum_{j=0}^{2^r-1}\cos^{N-|e|}\frac{\pi j}{2^r}\sin^{|e|}\frac{\pi j}{2^r}\sum_{q=0}^{2^r-1}(-1)^{f(q)}\Re\left(i^{|e|+q}e^{i\pi N/4}\right)\Re\left(e^{i\pi j(N-2q)/2^r}i^{|e|}\right).
\end{align}

\end{enumerate}
\end{proposition}
\begin{proof}
We write the hypergraph state as
\begin{equation}
\ket{H}=\frac{1}{\sqrt{2}^N}\sum_{ w=0}^N\sum_{|I|= w}(-1)^{f( w)}\ket{i_1\dots i_N},
\end{equation}
and use the identity
\begin{equation}
\sqrt{X}_+=\frac{1}{2}\left((1+i)\id+(1-i)X\right)=\frac{1+i}{2}\left(\id-iX\right).
\end{equation} 
Then, using $J=(j_1,\dots,j_N)$,
\begin{align*}
c_{|e|}&=\bra{\underbrace{1\dots 1}_{|e|} 0\dots 0}\sqrt{X}_+^{\otimes N}\ket{H}\\
&=\bra{1\dots 1 0\dots 0}\sum_{J\in\{0,1\}^N}\sum_ w\sum_{|I|= w}\left(\frac{1+i}{2}\right)^N\frac{(-1)^{f( w)}}{\sqrt{2}^N}(-iX)^{j_1}\otimes\dots\otimes(-iX)^{j_N}\ket{I}.
\end{align*}
Every $\ket{I}=\ket{i_1,\dots i_N}$ is mapped onto the subspace spanned by $\ket{1\dots 1 0\dots 0}$ by exactly one choice of $J$, specifically 
\begin{equation}\label{eq:choice_of_J}
j_k=
\begin{cases}
1-i_k &, k=1,\dots,|e|\\
i_k &, k=|e|+1,\dots,N.
\end{cases}
\end{equation}
By applying $(-iX)^{j_1}\otimes\dots\otimes(-iX)^{j_N}$ with this choice on $\ket{I}$, we pick up a total phase $(-i)^{|e|-(i_1+\dots+i_{|e|})}(-i)^{ w-(i_1+\dots+i_{|e|})}$ in the process. Substituting $m:=i_1+\dots+ i_{|e|}$, we get in total $\binom{N-|e|}{ w-m}\binom{|e|}{m}$ different possibilities to produce the factor $(-i)^{|e|+| w|-2m}$ in this fashion. Hence, we are left with
\begin{align}
\nonumber c_{|e|}&=\left(\frac{1+i}{2\sqrt{2}}\right)^N\sum_{ w}\sum_{|I|= w}(-1)^{f( w)}(-i)^{|e|+ w-2(i_1+\dots+i_{|e|})}\\
\label{eq:aux1}&=\left(\frac{1+i}{2\sqrt{2}}\right)^N\sum_{ w}\sum_{m=0}^{|e|}\binom{N-|e|}{ w-m}\binom{|e|}{m}(-1)^{f( w}(-i)^{|e|+ w-2m},
\end{align}
which proves $(i)$. We can proceed by separating the negative contributions as determined by $f( w)$:
\begin{align}
\label{eq:1/sqrt(2)N}c_{|e|}&=\left(\frac{1+i}{2\sqrt{2}}\right)^N\sum_{ w=0}^N\sum_{m=0}^{|e|}\binom{N-|e|}{ w-m}\binom{|e|}{m}(-i)^{|e|+ w-2m}\\
\label{eq:first_step}&-2\left(\frac{1+i}{2\sqrt{2}}\right)^N\sum_{f( w)=1}\sum_{m=0}^{|e|}\binom{N-|e|}{ w-m}\binom{|e|}{m}(-i)^{|e|+ w-2m}.
\end{align}
Omitting the factor $\left(\frac{1+i}{2\sqrt{2}}\right)^N$, we compute Eq.~\eqref{eq:1/sqrt(2)N}:
\begin{align*}
&\sum_{ w=0}^N\sum_{m=0}^{|e|}\binom{|e|}{m}\binom{N-|e|}{ w-m}(-i)^{|e|+ w-2m}=\sum_{m=0}^{|e|}\binom{|e|}{m}(-i)^{|e|-m}(1-i)^{N-|e|}\\
=&(1-i)^{N-|e|}(-i)^{|e|}\sum_{m=0}^{|e|}\binom{|e|}{m}i^m=(1-i)^{N-|e|}(-i)^{|e|}(1+i)^{|e|}\\
=&(1-i)^N.
\end{align*}
Therefore first summand is actually equal to $\frac{1}{\sqrt{2}^N}$, which reproduces the results of~\cite{gachechiladze2019quantum} (Lemma 5.9 therein).
Since $\ket{H}$ is stabilized by $X^{\otimes N}$, both sums in the first part are invariant under replacing $ w$ with $N- w$. This is sufficient to show that the expression in Eq.~\eqref{eq:aux1} is invariant under complex conjugation and therefore equal to its real part. This finishes the proof of the first claim.\\
Now assume that $f$ is indeed $2^r$ periodic with $r\geq 2$, cf.~\cite{Lyons2017Local}.
In order to arrive at Eq.~\eqref{eq:XCoefficientsGeomeasure2-1}, we rewrite Eq.~\eqref{eq:XCoefficientsGeomeasure-1} as
\begin{align*}
c_{|e|}=\frac{1}{(2\sqrt{2})^N} \sum_{q=0}^{2^r-1}(-1)^{f(q)}\sum_{m=0}^{|e|}\Re\left((1+i)^N(-i)^{q+|e|-2m}\right)\sum_{ w\equiv q\mod 2^r}\binom{N-|e|}{ w-m}\binom{|e|}{m}.
\end{align*}
Next we apply the first technical Lemma \ref{lemma:aux1} and reshuffle some of the terms:
\begin{align}
\nonumber c_{|e|}=&\frac{1}{2^N}\sum_{q=0}^{2^r-1}(-1)^{f(q)}\Re\left(e^{i\pi N/4 }(-i)^{q+|e|}\right)\sum_{m=0}^{|e|}\binom{|e|}{m}(-1)^m\times\\
&\frac{1}{2^r}\left(\sum_{j=0}^{2^r-1}(2\cos\frac{\pi j}{2^r})^{N-|e|}\cos\frac{\pi j(N-|e|+2m-2q)}{2^r}\right).
\end{align}
We already carried out the summation over $m$ in Lemma~\ref{lemma:aux2}, which then results in:
\begin{align*}
c_{|e|}=&\frac{1}{2^N}\sum_{q=0}^{2^r-1}(-1)^{f(q)}\Re\left(e^{i\pi N/4}(-i)^{q+|e|}\right)\times\\
&\frac{1}{2^r}\sum_{j=0}^{2^r-1}(2\cos\frac{\pi j}{2^r})^{N-|e|}\Re\left(e^{i\pi j(N-2q)/2^r}(-2i\sin\frac{\pi j}{2^r})^{|e|}\right).
\end{align*}
This indeed simplifies to \eqref{eq:XCoefficientsGeomeasure2-1}. 
\end{proof}

A similar result holds for the action of $\sqrt{Y}^{\otimes}$ on any $N$-qubit symmetric hypergraph state.
\begin{proposition}\label{prop:Yformula}
\begin{enumerate}
\item
Given a symmetric hypergraph state, then it is possible to calculate a coefficient of the computational basis element of a weight $|e|$ after application of $\sqrt{Y}^{\otimes N}$  using the following expression:
\begin{align}\label{eq:YCoefficientsGeomeasure}
c_{|e|}=\frac{(1+i)^N}{(2\sqrt{2})^N}\sum_{ w}(-1)^{f( w)+ w}\sum_{m=0}^{|e|}\binom{N-|e|}{ w-m}\binom{|e|}{m}(-1)^m,
\end{align}
where $f:\{0,\dots,N\}\to\{0,1\}$ again specifies the signs of the different weights.
\item
Again, if $f$ is $2^r$-periodic, this can alternatively be expressed as follows:
\begin{align}\label{eq:Ycoefficients_simplified}
c_{|e|}=\frac{1}{2^r}\sum_{j=0}^{2^r-1}\cos^{N-|e|}\frac{\pi j}{2^r}\sin^{|e|}\frac{\pi j}{2^r}\sum_{q=0}^{2^r-1}(-1)^{f(q)+q}\Re\left((-i)^{|e|}e^{i\pi j(N-2q)/2^r}\right).
\end{align}
\end{enumerate}
\end{proposition}
\begin{proof}
Similarly as before, we write
\begin{align*}
c_{|e|}&=\bra{1\dots 1 0\dots 0}\sum_{J\in\{0,1\}}^N\sum_ w\sum_{|I|= w}\left(\frac{1+i}{2}\right)^N\frac{(-1)^{f( w)}}{\sqrt{2}^N}(-iY)^{j_1}\otimes\dots\otimes(-iY)^{j_N}\ket{i_1\dots i_N}
\end{align*}
With the same choice of $J$ as in Eq.~\eqref{eq:choice_of_J}, keeping in mind that $-iY\ket{0}=\ket{1}, -iY\ket{1}=-\ket{0}$, we collect a factor of $1^{|e|}(-1)^{i_{|e|+1}+\dots+i_N}$. With the same substitution as before we obtain
\begin{equation}\label{eq:Yformula}
c_{|e|}=\left(\frac{1+i}{2\sqrt{2}}\right)^N\sum_ w(-1)^{f( w)}\sum_{m=0}^{|e|}\binom{N-|e|}{ w-m}\binom{|e|}{m}(-1)^{ w-m}.
\end{equation}
For the second part we proceed in the same manner as before, utilising the Lemmata~\ref{lemma:aux1}, \ref{lemma:aux2}.
\end{proof}

\subsection{The state transformation results for the $3$-uniform case}\label{proof:3-uniform}
\begin{replemma}{thm:theorem1}
The even-qubit $3$-uniform complete hypergraph states can be mapped to a superposition of the GHZ state and all odd weight vectors after applications of square roots of local Pauli matrices, corresponding to the respective stabilizers.
\vspace{5pt}\\
$\mathbf{N\equiv 2\mod 4}$: $\ket{H^3_N}$ is stabilized by $X^{\otimes N}$ and with $\ket{\tilde{H}^3_N}=\sqrt{X}_+^{\otimes N}\ket{H^3_N}$ we have 
\begin{equation}\label{eq:XthreeUnif_rep}
\ket{\tilde{H}_N^3}=\pm\frac{1}{\sqrt{2}}\ket{GHZ}+\frac{1}{\sqrt{2}^N}\sum_{w(x)\, \mathrm{odd}}\ket{x}.
\end{equation}
$\mathbf{N\equiv 0\mod 4}$: $\ket{H^3_N}$ is stabilized by $Y^{\otimes N}$ and with $\ket{\tilde{H}^3_N}=\sqrt{Y}_+^{\otimes N}\ket{H^3_N}$, we have
\begin{equation}\label{eq:YthreeUnif_rep}
\ket{\tilde{H}_N^3}=\pm\frac{1}{\sqrt{2}}\ket{GHZ}+\frac{1}{\sqrt{2}^N}\sum_{w(x)\, \mathrm{odd}}(-i)^{w(x)-1}\ket{x}.
\end{equation}

\end{replemma}
\begin{proof}
In order to prove the claim, one could proceed using the results from Propositions ~\ref{prop:Xformula} and \ref{prop:Yformula}. However, the $3$-uniform case allows for much neater proofs.
In both cases, the computational basis elements which have negative contributions are the ones with weight $w\equiv 3\mod 4$. Therefore, we can rewrite the state as
\begin{align*}
\ket{H^3_N}=&\frac{1}{\sqrt{2}^N}\left(\sum_{w(x)\,\mathrm{even}}\ket{x}+\sum_{w(x)\,\mathrm{odd}}i^{w(x')-1}\ket{x'}\right)\\
=&\frac{1}{2}\left(\ket{+}^{\otimes N}+\ket{-}^{\otimes N}\right)-\frac{i}{2}\left(\ket{+_Y}^{\otimes N}-\ket{-_Y}^{\otimes N}\right),
\end{align*}
Which is the same as having $\ket{H^3_N}=\frac{1}{\sqrt{2}}\ket{GHZ_X^+}+\frac{1}{\sqrt{2}}\ket{GHZ_Y^-}$.\\
\underline{Case 1:} $N\equiv 2\mod 4$:\\
Since $\ket{\pm}$ are the $\pm 1$ eigenstates of $X$, and $\sqrt{X}_+\ket{+_Y}=\frac{1+i}{\sqrt{2}}\ket{0},\,\sqrt{X}_+\ket{-_Y}=\frac{1-i}{\sqrt{2}}\ket{1}$, it is clear that
\begin{align*}
\sqrt{X}^{\otimes N}\ket{H^3_N}&=\frac{1}{\sqrt{2}}\sqrt{X}_+^{\otimes N}\left(\ket{+}^{\otimes N}+\ket{-}^{\otimes N}\right)+\frac{(-i)}{\sqrt{2}}\sqrt{X}_+^{\otimes N}\left(\ket{+_Y}^{\otimes N}-\ket{-_Y}^{\otimes N}\right)\\
&=\frac{1}{\sqrt{2}}\left(\ket{+}^{\otimes N}+i^N\ket{-}^{\otimes N}\right)+\frac{(-i)}{\sqrt{2}}\left(\left(\frac{1+i}{\sqrt{2}}\ket{0}\right)^{\otimes N}-\left(\frac{1-i}{\sqrt{2}}\ket{1}\right)^{\otimes N}\right).
\end{align*}
And since $N\equiv 2\mod 4$, this becomes
\begin{equation*}
=\frac{1}{\sqrt{2}}\left(\ket{+}^{\otimes N}-\ket{-}^{\otimes N}+(-i)(\pm i\ket{0}^{\otimes N}-(\mp i)\ket{1}^{\otimes N})\right)
\end{equation*}
\begin{equation*}
=\frac{1}{\sqrt{2}}\ket{GHZ_X^-}\pm\frac{1}{\sqrt{2}}\ket{GHZ}=\pm\frac{1}{\sqrt{2}}\ket{GHZ}+\frac{1}{\sqrt{2}^N}\sum_{w(x)\,\mathrm{odd}}\ket{x}.
\end{equation*}
\noindent \underline{Case 2:} $N\equiv 0\mod 4$\\
On the other hand, $\ket{\pm_Y}$ are the $\pm 1$ eigenstates of $Y$, further $\sqrt{Y}_+\ket{+}=\frac{1+i}{\sqrt{2}}\ket{1},\,\sqrt{Y}_+\ket{-}=\frac{1+i}{\sqrt{2}}\ket{0}$, so we get
\begin{align*}
\sqrt{Y}_+^{\otimes N}\ket{H^3_N}&=\frac{1}{\sqrt{2}}\sqrt{Y}_+^{\otimes N}\left(\ket{+}^{\otimes N}+\ket{-}^{\otimes N}\right)+\frac{(-i)}{\sqrt{2}}\sqrt{Y}_+^{\otimes N}\left(\ket{+_Y}^{\otimes N}-\ket{-_Y}^{\otimes N}\right)\\
&=\frac{1}{\sqrt{2}}\left(\left(\frac{1+i}{\sqrt{2}}\ket{1}\right)^{\otimes N}+\left(\frac{1+i}{\sqrt{2}}\ket{0}\right)^{\otimes N}\right)+\frac{(-i)}{\sqrt{2}}\left(\ket{+_Y}^{\otimes N}-i^N\ket{-_Y}^{\otimes N}\right)
\intertext{ and by assumption on $N$:}
&=\pm\frac{1}{\sqrt{2}}\ket{GHZ}+\frac{(-i)}{\sqrt{2}}\ket{GHZ_Y^-}.
\end{align*}
By applying $\sqrt{Z}_-$ on every site, we can transform this even further to
\begin{equation}
\pm\frac{1}{\sqrt{2}}\ket{GHZ}+\frac{i}{\sqrt{2}^N}\sum_{w(x)\,\mathrm{odd}}\ket{x}.
\end{equation}
The sign in front of the GHZ part depends on the exact value of $\left(\frac{1+i}{\sqrt{2}}\right)^N$ in both cases.
\end{proof}

\begin{lemma}\label{lemma:3,2-uniform}
Consider an even-qubit $(3,2)$-uniform hypergraph state $\ket{H^{3,2}_N}$. It can be written as
\begin{equation}
\ket{H^{3,2}_N}=\frac{1}{\sqrt{2}}\ket{GHZ_Y^+}+\frac{1}{\sqrt{2}}\ket{GHZ_X^-}.
\end{equation}
With a similar calculation as above, we see the following:\\
\underline{Case 1: $N=4k$}\\
The state is stabilized by $X^{\otimes N}$ and gets transformed to
\begin{align}
\nonumber\sqrt{X}_+^{\otimes N}\ket{H^{3,2}_N}&=\frac{1}{\sqrt{2}}\left(\left(\frac{1+i}{\sqrt{2}}\right)^N\ket{0}^{\otimes N}+\left(\frac{1-i}{\sqrt{2}}\right)^N\ket{1}^{\otimes N}+(\ket{+}^{\otimes N}-i^N\ket{-}^{\otimes N})\right)\\
\label{eq:3,2-uniformX}&=(-1)^{k}\frac{1}{\sqrt{2}}\ket{GHZ_Z^+}+\frac{1}{\sqrt{2}}\ket{GHZ_X^-}.
\end{align}
\underline{Case 2: $N=4k+2$}\\
The state ist stabilized by $Y^{\otimes N}$ and gets transformed to
\begin{align}
\nonumber\sqrt{Y}^{\otimes N}\ket{H^{3,2}_N}&=\frac{1}{\sqrt{2}}\left(\ket{+_Y}^{\otimes N}+i^N\ket{-_Y}^{\otimes N}+\left(\frac{1+i}{\sqrt{2}}\right)^N(\ket{1}^{\otimes N}-\ket{0}^{\otimes N})\right)\\
\label{eq:3,2-uniform}&=(-1)^{k+1}\frac{i}{\sqrt{2}}\ket{GHZ_Z^-}+\frac{1}{\sqrt{2}}\ket{GHZ_Y^-}.
\end{align}
In order for the second case to meet the requirements of Thm.~\ref{thm:General_Bell_violation}, we can apply local $\sqrt{Z}_+$ on every site and get rid of the relative phase in $\ket{GHZ_Z^-}$.
\end{lemma}

\subsection{Computations for coefficients beyond the $3$-uniform case}\label{proofs:general-uniform}
Interlude:\\
The proofs of Lemmata \ref{lemma:Xcoeff_general}, \ref{lemma:Ycoeff_general} both require a similar calculation, which we jointly conduct in the following:
Let $\sigma\in\{0,-1\}$, $r\geq 3$ and $f(q):=\binom{q}{2^{r-1}+1}\mod 2$. Note that $f(q)=1$ if and only if $q\in\{2^{r-1}+1,\dots,2^r-1\}$ and zero otherwise. This can be easily derived using Lucas' Theorem, which is stated e.g. in the Appendix of~\cite{Lyons2017Local}. Then:
\begin{align}
\nonumber\sum_{q=1\, \mathrm{odd}}^{2^r-1}(-1)^{f(q)}\left(i^\sigma e^{-2i\pi j/2^r}\right)^q=&\bigg(1-\left(i^\sigma e^{-2i\pi j/2^r}\right)^{2^{r-1}}\bigg)\sum_{q=1\, \mathrm{odd}}^{2^{r-1}-1}\left(i^\sigma e^{-2i\pi j/2^r}\right)^q\\
\nonumber=&(1-e^{i\pi j})i^\sigma e^{-2i\pi j/2^r}\sum_{q'=0}^{2^{r-2}-1}\left(i^\sigma e^{-4i\pi j/2^r}\right)^q.\\
\intertext{Clearly, this vanishes whenever $j$ is even. We therefore continue, assuming that $j$ is odd:}
\begin{split}\label{eq:aux3}
\sum_{q=1\, \mathrm{odd}}^{2^r-1}(-1)^{f(q)}\left(i^\sigma e^{-2i\pi j/2^r}\right)^q=&2i^\sigma e^{-2i\pi j/2^r}\frac{1-e^{i\pi j}}{1-(-1)^\sigma e^{-4i\pi j/2^r}}\\
=&\begin{cases}
-2i\left(\sin\frac{2\pi j}{2^r}\right)^{-1}, & \sigma=0\\
-2i\left(\cos\frac{2\pi j}{2^r}\right)^{-1}, & \sigma=-1.
\end{cases}
\end{split}
\end{align}

\subsubsection{Proof of Lemma \ref{lemma:Xcoeff_general}}\label{proof:lemma_X_coeff_general}
\begin{replemma}{lemma:Xcoeff_general}
Let $r\geq 3$ and $N=l2^r+2^{r-1}$. Then the $k=(2^{r-1}+1)$-uniform complete $N$-qubit hypergraph state is stabilized by $X^{\otimes N}$ and we have
\begin{equation}
\ket{\tilde{H}^k_N}=\sqrt{X}_+^{\otimes N}\ket{H^k_N}=\frac{1}{\sqrt{2}}\ket{GHZ_X}+(-1)^l\sum_{w(x)\, \mathrm{odd}} c_{w(x)}\ket{x},
\end{equation}
where $\ket{GHZ_X}=\frac{1}{\sqrt{2}}(\ket{+}^{\otimes N}+\ket{-}^{\otimes N})$, and the odd coefficients are given by
\begin{equation}\label{eq:Xcoeff_general_odd_rep}
c_w=\frac{2}{2^{r}}\sum_{j=1,\,\mathrm{odd}}^{2^r-1}i^{j-1}\frac{\cos\left(\frac{\pi j}{2^r}\right)^{N-w}\sin\left(\frac{\pi j}{2^r}\right)^{w}}{\cos\left(\frac{2\pi j}{2^{r}}\right)}.
\end{equation}
\end{replemma}
\begin{proof}
Recall that $k=2^{-1}r+1$ and $N=l2^r+2^{r-1}$. Because $N\equiv 0\mod 4$, and only odd weights $ w$ have negative contributions, using \eqref{eq:XCoefficientsGeomeasure-2} is easy to check that for even $|e|$, we have $\Re\left(e^{i\pi N/4}(-i)^{ w+|e|-2m}\right)=0$ and therefore $c_{|e|}=\frac{1}{\sqrt{2}^N}$. For odd $|e|$, consider the formula Eq.~\eqref{eq:XCoefficientsGeomeasure2-1}:
\begin{align}
\label{eq:aux10}c_{|e|}&=\frac{e^{i\pi N/4}}{2^r}\sum_{j=0}^{2^r-1}\cos^{n-|e|}\frac{\pi j}{2^r}\sin^{|e|}\frac{\pi j}{2^r}\sum_{q=0}^{2^r-1}(-1)^{f(q)}\Re((-i)^{|e|+q})\Re((-i)^{|e|}e^{i\pi j(N-2q)/2^r}).
\end{align}
When $q$ is even, we have $\Re((-i)^{|e|+q})=0$, and for odd $q$ we can rewrite 
\begin{align*}
&\Re((-i)^{|e|-q})\Re((-i)^{|e|}e^{i\pi j (N-2q)/2^r})=\Re((-i)^{2|e|-q}e^{i\pi j (l2^r+2^{r-1}-2q)/2^r}\\
=&\Re\left((-i)^q \underbrace{e^{i\pi j l}}_{=(-1)^{jl}}e^{-2i\pi j q/2^r}\underbrace{e^{i\pi j/2}}_{=i^j}\right).
\end{align*}
Therefore Eq.~\eqref{eq:aux10} can be rearranged to
\begin{align*}
&=\sum_{j=0}^{2^r-1}\frac{(-1)^{N/4+1}}{2^r}\cos\frac{\pi j}{2^r}^{N-|e|}\sin\frac{\pi j}{2^r}^{|e|}(-1)^{lj}\Re\bigg(\sum_{q=1\, \mathrm{odd}}^{2^r-1}(-1)^{f(q)}(-i)^{q-j}e^{-2i\pi jq/2^r}\bigg).
\end{align*}
Now we are in the situation of the preceeding remark and Eq.~\eqref{eq:aux3} in the case where $\sigma=-1$. Therefore, the only terms contributing are the ones where $j$ is odd, in which case we get
\begin{align}
\Re\left(i^j\sum_{q=1\, \mathrm{odd}}^{2^r-1} (-1)^{f(q)}(-i)^{q}e^{-2i\pi jq/2^r}\right)=\Re\left(i^j\frac{-2i}{\cos(\frac{2\pi j}{2^r})}\right)=\frac{2 i^{j-1}}{\cos(\frac{2\pi j}{2^r})}.
\end{align}
Putting everything together, we are left with
\begin{equation}
c_{|e|}=(-1)^{N/4+l}\frac{2}{2^r}\sum_{j=1\, \mathrm{odd}}^{2^r-1}\cos^{N-|e|}\frac{\pi j}{2^r}\sin^{|e|}\frac{\pi j}{2^r}\frac{i^{j-1}}{\cos\frac{2\pi j}{2^r}}.
\end{equation}
\end{proof}

\subsubsection{Proof of Lemma \ref{thm:Ygeneral} and Lemma \ref{lemma:Ycoeff_general}}\label{proof:lemma_Y_coeff_general}
\begin{replemma}{thm:Ygeneral}
Let $\ket{H^{\mathbf{k}}_N}=\frac{1}{\sqrt{2}^N}\sum_{x}(-1)^{f(w(x))}\ket{x}$ be a symmetric $\mathbf{k}$-uniform $N$-qubit hypergraph state, which is stabilized by $Y^{\otimes N}$. Further, assume that $f$ is $2^r$-periodic, $f(w)=0$ for all even $w$, and $N\equiv 0\mod 2^{r-1}$. Then the state can be mapped to a superposition of the GHZ state and some odd weights by applying $\sqrt{Y}_+^{\otimes N}$, i.e.
\begin{equation}\label{eq:Ygeneral_rep}
\ket{\widetilde{H}}=\sqrt{Y}_+^{\otimes N}\ket{H}=\frac{1}{\sqrt{2}}\ket{GHZ}+\frac{1}{\sqrt{2}}\ket{\phi_{\mathrm{odd}}^\mathbf{k}},
\end{equation} 
where $\ket{\phi_{\mathrm{odd}}^\mathbf{k}}$ is a normalized quantum state depending on $\mathbf{k}$ which features odd weight contributions only in the computational basis.
\end{replemma}

\begin{replemma}{lemma:Ycoeff_general}
Let $r\geq 3$ and $N=l2^r$. Then the $k=(2^{r-1}+1)$-uniform complete $N$-qubit hypergraph state is stabilized by $Y^{\otimes N}$ and gets mapped to
\begin{equation}
\ket{\tilde{H}^k_N}=\frac{1}{\sqrt{2}}\ket{GHZ}+\sum_{w(x)\, \mathrm{odd}}i^{w-1}c_{w(x)}\ket{x}
\end{equation}
by $\sqrt{Y}_+^{\otimes N}$, where the odd coefficients are given by
\begin{equation}
c_w=\frac{2}{2^{r}}\sum_{j=1,\,\mathrm{odd}}^{2^r-1}\frac{\cos(\frac{\pi j}{2^r})^{N-w}\sin(\frac{\pi j}{2^r})^{w}}{\sin(\frac{2\pi j}{2^{r}})}.
\end{equation}
\end{replemma}

\begin{proof}
First, we focus on the proof of Lemma \ref{lemma:Ycoeff_general}, i.e. $k=(2^{r-1}+1)$ and $N=l2^{r}$. With some minor alterations, these calculations can be adapted to also prove Lemma \ref{thm:Ygeneral}.
The weights in $\ket{H^k_N}$ with a negative sign are the odd weights congruent to $2^{r-1}+1,\dots,2^r-1$ modulo $2^r$. Assuming $r\geq 3$, we also have $\left(\frac{1+i}{\sqrt{2}}\right)^N=1$.
When $|e|=0$ or $|e|=N$, the expression \eqref{eq:YCoefficientsGeomeasure} reduces to
\begin{align*}
c_{|e|}&=\frac{1}{2^N}\left(\sum_{ w\,\mathrm{even}} (-1)^{f( w)+ w}\binom{N}{ w}\pm\sum_{ w\,\mathrm{odd}}(-1)^{f( w)+ w}\right)\\
&=\frac{1}{2^N}\left(\sum_{ w\, \mathrm{even}}\binom{N}{ w}\pm\left(\sum_{f( w)=0,\, \mathrm{odd}}\binom{N}{ w}-\sum_{f( w)=1,\,\mathrm{odd}}\binom{N}{ w}\right)\right).
\end{align*}
where the $+$ appears when $|e|=0$, and the $-$ for $|e|=N$. Let us fix an odd $0\leq  w\leq N$. According to the palindrome conditions in Lemma~\ref{lemma:localStab} we have $f( w)\equiv f(N- w)+1\mod 2$ as well as $\binom{N}{ w}=\binom{N}{N- w}$. Hence the odd contributions cancel each other, which leaves us with
\begin{equation}
c_0=c_N=\frac{1}{2}
\end{equation}
We calculate the other coefficients. Similar to the proof of Lemma \ref{thm:Ygeneral}, the formula \eqref{eq:Ycoefficients_simplified} can be rearranged to
\begin{align}
\label{eq:theothercase} c_{|e|}&=\frac{1}{2^{r}}\sum_{j=0}^{2^r-1}\cos^{N-|e|}\left(\frac{\pi j}{2^{r}}\right)\sin^{|e|}\left(\frac{\pi j}{2^{r}}\right)\sum_{q=0}^{2^{r}-1}(-1)^{f(q)+q}(-1)^{lj}\Re\left(e^{-2i q j\pi/2^{r}}(-i)^{|e|}\right)
\end{align}
For $j=0$, we have clearly no contribution since $|e|>0$. Let us therefore assume that $j\neq 0$. It is easily seen that in this case the contributions where $q$ is even sum up to zero:
\begin{align}
\sum_{q\,\mathrm{even}} (-1)^{f(q)}e^{-2i\pi j q/2^r}=\sum_{q'=0}^{2^{r-1}}e^{-4i\pi j q'/2^r}=0.
\end{align}
Now we employ Eq.~\eqref{eq:aux3} in the $\sigma=0$ case, to conclude that we still have $\sum_{q=1,\, \mathrm{odd}}^{2^r-1} (-1)^{f(q)+q}e^{-2i\pi j q/2^r}=0$, when $j$ is even and
\begin{align}
\label{eq:realpartzero}\Re\left(\sum_{q=1\, \mathrm{odd}}^{2^r-1} (-1)^{f(q)+q}e^{-2i\pi j q/2^r}\right)=-\Re\left((-i)^{|e|}\frac{-2i}{\sin(\frac{2\pi j}{2^r})}\right)
= 2\Re\left(\frac{(-i)^{|e|-1}}{\sin(\frac{2\pi j}{2^r})}\right),
\end{align}
when $j$ is odd. Clearly this expression equals zero whenever $|e|$ is even. If $|e|$ is odd on the other hand, combining everything we get
\begin{align}
c_{|e|}=\frac{2}{2^{r}}(-1)^{l}(-i)^{|e|-1}\sum_{j=1,\, \mathrm{odd}}^{2^r-1}\frac{\cos^{N-|e|}\left(\frac{\pi j}{2^{r}}\right)\sin^{|e|}\left(\frac{\pi j}{2^{r}}\right)}{\sin\left(\frac{2\pi j}{2^r}\right)}.
\end{align}
We can apply local $\sqrt{Z}^{\otimes N}_\pm$ operators, in order to get rid of the alternating phases on the odd weight contributions, at the price of picking up an imaginary unit. Neglecting global phases, we can then transform the state as claimed, proving Lemma \ref{lemma:Ycoeff_general}.

Lastly, we want to transfer the qualitative result that all even weights apart from $c_0=c_N=\frac{1}{2}$ are zero, to the more general assumptions in Lemma \ref{thm:Ygeneral}. First of all, according to Lemma \ref{lemma:localStab} and the fact that $N\equiv 0\mod 4$ by our assumption, we have that $\mathbf{k}=(k_1,\dots,k_n$) must satisfy  
\begin{equation}
\sum_i \binom{ w}{k_i}\equiv\sum_i \binom{N- w}{k_i}+ w\mod 2\textup{ for all }0\leq w\leq N.
\end{equation}
By assumption, the terms where the binomials $\binom{N}{w}$ are congruent $1$ modulo $2$ are only distributed only among the odd weights, and the contributions of weights $ w$ and $N- w$ must have opposite signs if $ w$ is odd. Hence we can conclude right away that $c_{|e|}=\frac{1}{2}$ when $|e|\in\{0,N\}$. In all the other cases where $|e|$ is even, we can assume without loss of generality, that the negative weights only occur when $ w\mod 2^r \in\{2^{r-1}+1,\dots,2^{r-1}+3,\dots,2^r-1\}$. If necessary, we can exchange $ w$ with $N- w$, since with $m'=m-|e|$
\begin{equation}
\sum_m^{|e|}\binom{N-|e|}{N- w-m}\binom{|e|}{m}(-1)^{m}=\sum_{m'}^{|e|}\binom{N-|e|}{ w-m'}\binom{|e|}{m'}(-1)^{m'+|e|}.
\end{equation}
Let therefore $|e|$ be even and fixed. As we just argued,
\begin{equation}
c_{|e|}=\frac{2(1+i)^N}{(2\sqrt{2})^N}\sum_{f( w)=1}\sum_{m=0}^{|e|}(-1)^m\binom{N-|e|}{ w-m}\binom{|e|}{m}
\end{equation}
is therefore actually equal to
\begin{equation}
c_{|e|}=\frac{2(1+i)^N}{(2\sqrt{2})^N}\sum_{\stackrel{ w\equiv 2^{r-1}+1}{\mod 2^r}}^{2^r-1}\sum_{m=0}^{|e|}(-1)^m\binom{N-|e|}{ w-m}\binom{|e|}{m},
\end{equation}
and we are in a similar situation as before. When $N\equiv 0\mod 2^r$, the computation works out in exactly the same manner as above. 
\end{proof}

\subsection{Analytical results on Geometric entanglement measure}\label{sec:GeoMeasure}
\begin{lemma}\label{lemma:cyclic_max}
For $r\geq 1$ and even $N$, the function 
\begin{equation}
f(x)=\sum_{j=0}^{2^r-1}\cos^N\left(x+\frac{\pi j}{2^r}\right)
\end{equation}
attains its maximum at $x\in \frac{\pi}{2^r}\Z$
\end{lemma}
\begin{proof}
Due to the symmetry, it suffices to check that $f$ attains its maximum on $[-\pi/2^r,\pi/2^r]$ at $x=0$.
Here we can exploit the following, using that $N$ is even, see~\cite{Whitaker} p. 263:
\begin{align*}
\cos^N(x)=\frac{1}{2^{N-1}}\frac{N!}{((N/2)!)^2}\left(\frac{1}{2}+\frac{N}{N+2}\cos(2x)+\frac{N(N-2)}{(N+2)(N+4)}\cos(4x)+\dots\right).
\end{align*} 
For odd $l$, we have that
\begin{align}\label{eq:aux11}
\sum_{j=0}^{2^r-1}\cos\left(2l(x+\frac{\pi j}{2^r})\right)=\sum_{j=0}^{2^{r-1}-1}\cos\left(2l(x+\frac{\pi j}{2^r})\right)+\cos\left(2l(x+\frac{\pi (j+2^{r-1})}{2^r})\right)=0,
\end{align}
and in the case where $l$ is even, we can write
\begin{equation*}
\sum_{j=0}^{2^r-1}\cos\left(2l(x+\frac{\pi j}{2^r})\right)=\sum_{j=0}^{2^r-1}\cos\left(l(2x+\frac{\pi j}{2^{r-1}})\right).
\end{equation*}
If $l/2$ is still even, we apply the same reduction again, otherwise this term vanishes, according to \eqref{eq:aux11}.
Inductively, we can show that for any $l$ not divisible by $2^{r}$, the sum over all the shifted contrubutions vanishes. For $l\equiv 0\mod 2^r$ on the other hand, the reduction terminates when we are only left with one summand. In this case, we get $\cos(2l(x+\frac{\pi j}{2^r}))=\cos(2lx)$ for all $j$ and therefore we can write
$f(\theta)=\sum_k\alpha_k\cos(2^{r+1}k\theta)$,
for some suitable, nonnegative coefficients $\alpha_k$. This indeed shows, that $f$ attains its maximum at $\theta=0$ and by symmetry also at every $\theta\in\frac{\pi}{2^r}$.
\end{proof}

\subsubsection{Geometric entanglement in the $5$-uniform case}\label{proof:5-uniform_GEM}
With some tricks, the strategy for proving the $3$-uniform case can be applied to the $5$-uniform $X^{\otimes N}$-stabilized case. Recall that the odd coefficients are given by
\begin{equation}
c_w=\frac{1}{4}\sum_{j=1,\, \mathrm{odd}}^{7}i^{j-1}\frac{\cos\left(\frac{\pi j}{8}\right)^{N-w}\sin\left(\frac{\pi j}{8}\right)^{w}}{\cos\left(\frac{\pi j}{4}\right)},
\end{equation}
which is positive for all odd $j$ and $w$. Without loss of generality, we assume that the the even weight coefficients have positive signs as well. We compute
\begin{align*}
f_N(\theta):=&\braket{\tilde{H}^5_N}{\psi(\theta)}^{\otimes N}\\
=&\frac{1}{2}\cos^N\left(\theta+\frac{\pi}{4}\right)+\frac{1}{2}\cos^N\left(\theta+\frac{3\pi}{4}\right)\\
+&\frac{1}{4}\sum_{j=1,\, \mathrm{odd}}^7\frac{i^{j-1}}{\cos(\frac{\pi j}{4})}\sum_{k=1\, \mathrm{odd}}^N\binom{N}{k}\cos^{N-k}(\theta)\cos^{N-k}(\frac{\pi j}{8})\sin^k(\theta)\sin^k(\frac{\pi j}{8}).
\end{align*}
This can be simplified to
\begin{align*}
f_N(\theta)=&\frac{1}{2}\sum_{j=2,6}\cos^N\left(\theta+\frac{\pi j}{8}\right)+\frac{1}{8}\sum_{j=1,\, \mathrm{odd}}^7\frac{i^{j-1}}{\cos(\frac{\pi j}{4})}\left(\cos^N\left(\theta-\frac{\pi j}{8}\right)-\cos^N\left(\theta+\frac{\pi j}{8}\right)\right)\\
=&\frac{1}{2}\sum_{j=2,6}\cos^N\left(\theta+\frac{\pi j}{8}\right)+\frac{1}{2\sqrt{2}}\left(\sum_{j=5,7}\cos^N\left(\theta+\frac{\pi j}{8}\right)-\sum_{j=1,3}\cos^N\left(\theta+\frac{\pi j}{8}\right)\right)\\
=&\left(\frac{1}{2}-\frac{1}{2\sqrt{2}}\right)\sum_{j=2,6}\cos^N\left(\theta+\frac{\pi j}{8}\right)+\frac{1}{2\sqrt{2}}\sum_{j=0}^7\cos^N\left(\theta+\frac{\pi j}{8}\right)\\
&\hspace{1cm}-\frac{1}{2\sqrt{2}}\sum_{j=0,4} \cos^N\left(\theta+\frac{\pi j}{8}\right)-\frac{2}{2\sqrt{2}}\sum_{j=1,3}\cos^N\left(\theta+\frac{\pi j}{8}\right).
\end{align*}
Again the sum running over all $j=0,\dots 7$ attains its maxima at $\frac{\pi}{8}\Z$, in particular at $\theta=\frac{\pi}{4}$. Also, the first sum is maximal at $\theta=\frac{\pi}{4}$, whereas the negative contributions are minimal at $\theta=\frac{\pi}{4}$. This is obvious for the first one $(j=0,4)$. For the second one, we can quickly check this analytically:
\begin{equation}
g(x):=\cos^N\left(x+\frac{\pi}{8}\right)+\cos^N\left(x+\frac{3\pi}{8}\right).
\end{equation}
We now look at the roots of the first derivative:
\begin{align*}
0\stackrel{!}{=}\del_x g(x)=-N\left(\cos^{N-1}\left(x+\frac{\pi}{8}\right)\sin\left(x+\frac{\pi}{8}\right)+\cos^{N-1}\left(x+\frac{3\pi}{8}\right)\sin\left(x+\frac{3\pi}{8}\right)\right)
\end{align*}
After dividing through $\cos^N(x+\frac{\pi}{8})$ and substituting $y:=\tan(x+\frac{\pi}{8})$, this is equivalent to demanding that
\begin{equation}\label{eq:equivalent_condition}
L(y):=\frac{-2y}{1-y^2}=\left(\frac{1-y}{\sqrt{2}}\right)^{N-2}=:R_N(y).
\end{equation}
For $x\in[-3\pi/8,-\pi/8]$ we have $  g(x)\geq\frac{2}{\sqrt{2}^N}\geq g(\pi/4)$ therefore, we can exclude solutions of Eq.~\eqref{eq:equivalent_condition} with $-1<y\leq 0$ from our search of the global minimum. When $y<-1$ or $0\leq y<1$, $L(y)$ and $R_N(y)$ have opposite signs (recall that $N$ is even), so we find no solutions in those intervals. For $y>1$, $L(y)$ is strictly decreasing and $R_N(y)$ strictly increasing. Hence there is at most one solution in this interval. This solution is found at $y=\sqrt{2}+1\Leftrightarrow x=\frac{\pi}{4}$ and easily verified to correspond to a local minimum, which therefore also must be the global minimum of $g$.
Alltogether, we can conclude that $f_N$ attains its maximum at $\theta=\frac{\pi}{4}$, which is then equal to
\begin{equation}
\max_\theta f_N(\theta)=\frac{1}{2}+\frac{1}{\sqrt{2}}\left(\cos^N\left(\frac{\pi}{8}\right)-\sin^N\left(\frac{\pi}{8}\right)\right).
\end{equation}

For $Y^{\otimes N}$-stabilized $5$-uniform states, again we can no longer assume all coefficients to be real, therfore, we need to introduce a complex phase. However, we can still assume that the closest state is of the form
\begin{equation}
\ket{\psi(\theta,\phi)}^{\otimes N}=(\sin(\theta)\ket{0}+e^{i\phi}\cos(\theta)\ket{1})^{\otimes N}
\end{equation}
for some $\theta\in[0,\frac{\pi}{2}]$.
We compute the overlap of $\ket{\tilde{H}^5_N}=\sqrt{Z}_+^{\otimes N}\sqrt{Y}_+^{\otimes N}\ket{H^5_N}$ with such a general symmetric state:
\begin{align*}
\braket{\tilde{H}^5_N}{\psi(\theta,\phi)}^{\otimes N}&=\frac{i}{2}\left(\sin^N(\theta)+e^{i\phi N}\cos^N(\theta)\right)\\
&+\frac{1}{4}\sum_{j=1,\, \mathrm{odd}}^7\sum_{k=1\, \mathrm{odd}}^N\binom{N}{k}\frac{\sin^k(\frac{\pi j}{8})\cos^{N-k}(\frac{\pi j}{8})}{\sin(\frac{\pi j}{4})}e^{i\phi(N-k)}\cos^{N-k}(\theta)\sin^k(\theta)
\end{align*}
The expression 
\begin{equation*}
\frac{\sin^k(\frac{\pi j}{8})\cos^{N-k}(\frac{\pi j}{8})}{\sin(\frac{\pi j}{4})}\cos^{N-k}(\theta)\sin^k(\theta)
\end{equation*}
is nonnegative for all odd $j\in\{1,\dots,2^r-1\}$, odd $k$ and $\theta\in[0,\frac{\pi}{2}]$, therfore we can estimate the modulus by
\begin{align*}
|\braket{\tilde{H}^5_N}{\psi(\theta,\phi)}^{\otimes N}|\leq&  \frac{1}{2}\sin^N(\theta)+\frac{1}{2}\cos^N(\theta) \\
+&\frac{1}{4}\sum_{j=1,\, \mathrm{odd}}^7\sum_{k=1\, \mathrm{odd}}^N\binom{N}{k}\frac{\sin^k(\frac{\pi j}{8})\cos^{N-k}(\frac{\pi j}{8})}{\sin(\frac{\pi j}{4})}\cos^{N-k}(\theta)\sin^k(\theta)\\
=&\frac{1}{2}\sum_{j=0,4}\cos^N\left(\theta+\frac{\pi j}{8}\right)\\
+&\frac{1}{2\sqrt{2}}\left(\sum_{j=5,7}\cos^N\left(\theta+\frac{\pi j}{8}\right)-\sum_{j=1,3}\cos^N\left(\theta+\frac{\pi j}{8}\right)\right).
\end{align*}
This does not get maximal neither at $\theta=\frac{\pi}{4}$ nor at $\theta=0$, which makes it hard to analytically determine its maximum explicitly. Therfore we again rely on estimates. For $N\geq 16$, it is not difficult to check that having $|\braket{\tilde{H}^5_N}{\psi(\theta,\phi)}|\geq\frac{1}{2}$ requires $\theta\in[-\frac{\pi}{8},\frac{\pi}{8}]+\frac{\pi}{2}\Z$. Without loss of generality assume that $\theta\in[-\frac{\pi}{8},\frac{\pi}{8}]$, which allows us to estimate
\begin{equation*}
\cos^N\Big(\theta+\frac{5\pi}{8}\Big)-\cos^N\left(\theta+\frac{\pi}{8}\right)\leq 0\leq \cos^N\left(\theta+\frac{\pi}{8}\right)-\cos^N\Big(\theta+\frac{5\pi}{8}\Big).
\end{equation*}
Thereby we can relate the maximization to the $X^{\otimes N}$-stabilized case:
\begin{align*}
\max_{\phi,\theta}|\braket{\tilde{H}^5_N}{\psi(\theta,\phi)}^{\otimes N}|&\leq\max_{\theta} \frac{1}{2}\sum_{j=0,4}\cos^N\left(\theta+\frac{\pi j}{8}\right)\\
&+\frac{1}{2\sqrt{2}}\left(\sum_{j=1,7}\cos^N\left(\theta+\frac{\pi j}{8}\right)-\sum_{j=3,5}\cos^N\left(\theta+\frac{\pi j}{8}\right)\right)\\
&=\max_{\theta}f_N\left(\theta+\frac{\pi}{4}\right)=\frac{1}{2}+\frac{1}{\sqrt{2}}\left(\cos^N\left(\frac{\pi}{8}\right)-\sin^N\left(\frac{\pi}{8}\right)\right),
\end{align*}
which proves the lower bound in \eqref{eq:5_unif_GEM_Y}. The upper bound follows trivially from the fact that $|\braket{\tilde{H}^5_N}{0}^{\otimes N}|=\frac{1}{2}$.
\newpage

\section{Appendix on nonlocality}\label{sec:nonlocality}
This part of the Appendix is dedicated to provide detailed calculations and reasoning for the statements of the main text concerning nonlocality and robustness, in this order.
\subsection{Computation for nonlocality of the X-stabilized case}\label{sec:Xviolation}
Let $k=2^{r-1}+1$ and $N\equiv 2^{r-1}\mod 2^r$, so that $\ket{H^k_N}$ is $X^{\otimes N}$-stabilized. Consider the Bell operator
\begin{align*}
\mathcal{B}^Y_N&=\frac{1}{2}(Y+iZ)^{\otimes N}+\frac{1}{2}(Y-iZ)^{\otimes N}\\
&=\sum_{m\textup{ even}}i^m\underbrace{Y_1\dots Y_m Z_{m+1}\dots Z_N}_{=:A(m)} +\mbox{perm}.
\end{align*}
Because $\mathcal{B}^Y_N$ is invariant under conjugation with $\sqrt{X}_+^{\otimes N}$, we directly see that
\begin{align*}
\bra{H^k_N}\mathcal{B}^Y_N\ket{H^k_N}=\bra{\tilde{H}^k_N}\mathcal{B}^Y_N\ket{\tilde{H}^k_N}=\frac{1}{2}\bra{GHZ_X}\mathcal{B}^Y_N\ket{GHZ_X}+\bra{\phi_{\mathrm{odd}}}\mathcal{B}^Y_N\ket{\phi_{\mathrm{odd}}}.
\end{align*}
Note that the cross-terms cancel, as the summands of $\mathcal{B}^Y_N$ contain only even numbers of bit-flips. Let us first consider the GHZ part. 
First, we observe that because $N\equiv 0\mod 4$, we can rewrite $\mathcal{B}^Y_N$ in the following way:
\begin{equation}
    \mathcal{B}^Y_N=\frac{1}{2}Z^{\otimes N}\left((\id+X)^{\otimes N}+(\id-X)^{\otimes N}\right)=\frac{1}{2}Z^{\otimes N}\tilde{\mathcal{B}}^X_N.
\end{equation}
Since $\ket{GHZ_X}$ is a $+1$-eigenstate for $Z^{\otimes N}$ and all summands of $\tilde{\mathcal{B}}^X_N$ - recall that these always contain an even number of Pauli-$X$ operators and identities otherwise - we clearly get
\begin{equation}
\frac{1}{2}\bra{GHZ_X}\mathcal{B}^Y_N\ket{GHZ_X}=\frac{1}{2}\bra{GHZ_X}Z^{\otimes N}\tilde{\mathcal{B}}^X_N\ket{GHZ_X}= 2^{N-2}.
\end{equation}

Let us consider the odd contributions next.
Lastly, we compute the contribution of the odd weight terms for $0\leq m\leq N$. For a single correlator $A(m)$, the contribution on the odd terms is given by
\begin{align}
\nonumber\frac{1}{2}\bra{\phi_{\mathrm{odd}}}A(m)\ket{\phi_{\mathrm{odd}}}=&\frac{1}{2}\bra{\phi_{\mathrm{odd}}}Y_1\dots Y_m Z_{m+1}\dots Z_N\ket{\phi_{\mathrm{odd}}}\\
\label{eq:aux12}=&\frac{1}{2}\sum_{w,w'\,\mathrm{odd}}\sum_{|I|=w}\sum_{|J|=w'}c_w c_{w'}\bra{I}Y_1\dots Y_m Z_{m+1}\dots Z_N\ket{J}
\end{align}
The only terms which contribute to this summation are the ones where $(j_1,\dots,j_N)=(\bar{i_1},\dots,\bar{i_m},i_{m+1},\dots,i_N)$. 
Also, $Y$ and $Z$ both introduce a sign, when acting on $\ket{1}$, so we pick up a total phase of $-i^m$ in the process. Further, we substitute $j:=i_1+\dots +i_m$, that is if $I$ has weight $w$, the corresponding $J$ has weight $w'=|J|=w+m-2j$. Alltogether we count $\binom{m}{j}\binom{N-w}{m-j}$ possibilities to distribute a weight of $j$ among the first $m$ entries, given that the total weight is $w$. Exploiting the symmetry, we can then reduce Eq.~\eqref{eq:aux12} to
\begin{align}
\label{eq:generalform}=-i^m&\sum_{w,\,\mathrm{odd}}\sum_{j=0}^m c_wc_{w+m-2j}\binom{m}{j}\binom{N-m}{w-j}.
\end{align}
In order to reduce cumbersome notation, let $a(l):=\cos(\frac{\pi l}{2^r})$ and $b(l):=\sin(\frac{\pi l}{2^r})$. We proceed calculating the product of the coefficients, using the general formulae derived for the odd coefficients in Lemma~\ref{lemma:Xcoeff_general}.
\begin{align}\label{eq:product}
c_wc_{w+m-2j}=\frac{4}{4^r}\sum_{\stackrel{l,l'=1}{\mathrm{odd}}}^{2^r-1}\frac{-i^{l+l'}}{a(2l)a(2l')}a(l)^{N-w}a(l')^{N-w-m+2j}b(l)^{w}b(l')^{w+m-2j}\qquad
\end{align}
For now, let us fix $l$ and $l'$, and compute the sum over $j$ and $w$:
\begin{align*}
&\sum_{j=0}^m\sum_{\stackrel{w'\equiv j+1}{\mod 2}}\binom{m}{j}\binom{N-m}{w'}a(l)^{N-w'-j}a(l')^{N-w'-m+j}b(l)^{w'+j}b(l')^{w'+m-j}\\
=&\sum_{j=0}^m\binom{m}{j}a(l)^{N-j}a(l')^{N-m+j}b(l)^{j}b(l')^{m-j}\times\\
&\hspace{.2\linewidth}\frac{1}{2}\left(\left(1+\frac{b(l)b(l')}{a(l)a(l')}\right)^{N-m}-(-1)^j\left(1-\frac{b(l)b(l')}{a(l)a(l')}\right)^{N-m}\right)\\
=&\frac{1}{2}\sum_{\sigma\in\{0,1\}} (-1)^\sigma a(l)^Na(l')^{N-m}b(l')^m\left(1+(-1)^\sigma\frac{b(l)b(l')}{a(l)a(l')}\right)^{N-m}\left(1+(-1)^\sigma\frac{a(l')b(l)}{a(l)b(l')}\right)^m\\
=&\frac{1}{2}\sum_{\sigma\in\{0,1\}} (-1)^\sigma \left(a(l)a(l')+(-1)^\sigma b(l)b(l')\right)^{N-m}\left(a(l)b(l')+(-1)^\sigma a(l')b(l)\right)^m\\
=&\frac{1}{2}\sum_{\sigma\in\{0,1\}} (-1)^\sigma\cos^{N-m}\left((l-(-1)^\sigma l')\frac{\pi}{2^r}\right)\sin^{m}\left((l+(-1)^\sigma l')\frac{\pi}{2^r}\right).
\end{align*}
By replacing $l'$ with $(2^r-l')$ in the $\sigma=1$-term, we see that the summation over $l,l'$ then gives the same contribution for $\sigma=0$ and $\sigma=1$.
Adding up the resulting contributions for the different $m$, we arrive at:
\begin{align*}\
\frac{1}{2}\bra{\phi_{\mathrm{odd}}}\mathcal{B}^Y_N\ket{\phi_{\mathrm{odd}}}=&-\sum_{\stackrel{m=0}{\mathrm{even}}}^N\binom{N}{m}i^m\sum_{\stackrel{w=1}{\mathrm{odd}}}^{N-1}i^m c_wc_{w+m-2j}\binom{m}{j}\binom{N-m}{w-j}\\
=&\frac{4}{4^r}\sum_{\stackrel{m=0}{\mathrm{even}}}^{N}\binom{N}{m}\sum_{\stackrel{l,l'=1,}{\mathrm{odd}}}^{2^r-1}i^{l+l'}\frac{\sin^m\left((l+l')\frac{\pi}{2^r}\right)\cos^{N-m}\left((l-l')\frac{\pi}{2^r}\right)}{\cos(\frac{2\pi l}{2^r})\cos(\frac{2\pi l'}{2^r})}\\
=&\frac{2}{4^r}\sum_{\stackrel{l,l'=1}{\mathrm{odd}}}^{2^r-1}\frac{i^{l+l'}}{\cos\left(\frac{2\pi l}{2^r}\right)\cos\left(\frac{2\pi l'}{2^r}\right)}\Bigg[\left(\cos\left(\frac{(l-l')\pi}{2^r}\right)+\sin\left(\frac{(l+l')\pi}{2^r}\right)\right)^N\\
&\hspace{0.28\linewidth}+\left(\cos\left(\frac{(l-l')\pi}{2^r}\right)-\sin\left(\frac{(l+l')\pi}{2^r}\right)\right)^N\Bigg].
\end{align*}
In the second summand, we exchange $l\leftrightarrow 2^r-l$, $l'\leftrightarrow 2^r-l'$ and thereby get rid of the relative minus sign. Thus we get
\begin{align*}
=&\frac{4}{4^r}\sum_{\stackrel{l,l'=1}{\mathrm{odd}}}^{2^r-1}\frac{i^{l+l'}}{\cos(\frac{2\pi l}{2^r})\cos(\frac{2\pi l'}{2^r})}\left(\cos(\frac{(l-l')\pi}{2^r})+\sin(\frac{(l+l')\pi}{2^r})\right)^N\\
=&\frac{4}{4^r}\sum_{\stackrel{l,l'=1}{\mathrm{odd}}}^{2^r-1}i^{l+l'}\frac{\left[\cos(\frac{l\pi}{2^r})\cos(\frac{l'\pi}{2^r})+\cos(\frac{l\pi}{2^r})\sin(\frac{l'\pi}{2^r})+\sin(\frac{l\pi}{2^r})\cos(\frac{l'\pi}{2^r})+\sin(\frac{l\pi}{2^r})\sin(\frac{l'\pi}{2^r})\right]^N}{\cos(\frac{2\pi l}{2^r})\cos(\frac{2\pi l'}{2^r})}\\
=&-\frac{4}{4^r}\sum_{l=1,\,\mathrm{ odd}}^{2^r-1} i^{l}\frac{\left(\cos(\frac{l\pi}{2^r})+\sin(\frac{l\pi}{2^r})\right)^N}{\cos(\frac{2\pi l}{2^r})}\sum_{l'=1,\,\mathrm{odd}}^{2^r-1}i^{-l'}\frac{\left(\cos(\frac{l'\pi}{2^r})+\sin(\frac{l'\pi}{2^r})\right)^N}{\cos(\frac{2\pi l'}{2^r})}\\
=&-\frac{4}{4^r}\left|\sum_{l=1,\,\mathrm{odd}}^{2^r-1}i^l\frac{\left(\cos(\frac{l\pi}{2^r})+\sin(\frac{l\pi}{2^r})\right)^N}{\cos(\frac{2l\pi}{2^r})}\right|^2,
\end{align*}
which finishes our derivation.

\subsection{Violation of the separability inequality in the $3$-uniform complete case in the presence of particle loss}\label{sec:particleloss}
The Bell operator for the $3$-uniform complete hypergraph state after losing $k$ particles:
\begin{align}\label{eq:MerminOperatorOur0}
\begin{split}
    \mathcal{M}^0_{N-k}&=  X_{1}X_{2}X_{3}\dots X_{N-k} \\
                 & -  Z_{1}Z_{2}X_{3}X_{4}X_{5}\dots X_{N-k} -  \mbox{ perm.} \\
                 & +  Z_{1}Z_{2}Z_{3}Z_{4}X_{5}X_{6}X_7\dots X_{N-k} + \mbox{ perm.}\\
                & - \dots,
\end{split}\\
\label{eq:MerminOperatorOur1}
\begin{split}
    \mathcal{M}^1_{N-k} &=  Z_{1}X_{2}X_{3}\dots X_{N-k} +   \mbox{ perm.} \\
                 & -  Z_{1}Z_{2}Z_{3}X_{4}X_{5}\dots X_{N-k} -  \mbox{ perm.}\\
                 & +  Z_{1}Z_{2}Z_{3}Z_{4}Z_{5}X_{6}X_7\dots X_{N-k} + \mbox{ perm.}\\
                & - \dots,
\end{split}
\end{align}
Additional to the already introduced notation $\ket{GHZ_X}=\frac{1}{\sqrt{2}}(\ket{+}^{\otimes N}+\ket{-}^{\otimes N})$ we write $\ket{GHZ_X^\pm}=\frac{1}{\sqrt{2}}(\ket{+}^{\otimes N}\pm\ket{-}^{\otimes N})$ and similarly for the GHZ state with respect to the Pauli-$Y$ eigenbasis.
\subsubsection{Proof of Theorem \ref{thm:robustness}}\label{prf:thm_robustness}
\begin{reptheorem}{thm:robustness}
After loosing $k$ particles of $N$-qubit $3$-uniform complete hypergraph state, we can derive the following violations of separability inequalities:\vspace{8pt}\\
    \begin{tabular}{c|c|c}
    \textup{Constraints on} $N$ \textup{and} $k$ & \textup{Bell inequality} $\mathcal{B}_{N\backslash k}$ & \textup{Quantum value}
    $\mean{\mathcal{B}_{N\backslash k}}_{\ket{H^3_{N}}}$\\
    \hline
    $N\equiv 2\mod 4$, $k$ odd & $\mathcal{M}^0_{N-k}\otimes\id^{\otimes k}$ & $\sqrt{2}^{N-2k}$\\
    $N\equiv 0\mod 4$, $k$ even &$\mathcal{M}^1_{N-k}\otimes\id^{\otimes k}$ & $\sqrt{2}^{N-2k}$\\
    $N-k\equiv 2\mod 4$ & $\mathcal{M}^1_{N-k}\otimes\id^{\otimes k}$ & $\left|\sin\left(\frac{\pi k}{4}\right)\right|\sqrt{2}^{N-2k}$\\
    $N-k\equiv 0\mod 4$ & $\mathcal{M}^1_{N-k}\otimes\id^{\otimes k}$ & $\left|\cos\left(\frac{\pi k}{4}\right)\right|\sqrt{2}^{N-2k}$
\end{tabular}\vspace{5pt}\\
In the first two cases, the quantum value is $1/2$ and $0$, if instead we consider even and odd $k$, respectively.
\end{reptheorem}
\begin{proof}
\underline{Case 1: $N\equiv 2\mod 4$}\\
We transform the operator $\mathcal{B}_{N\backslash k}=\mathcal{M}^0_{N-k}\otimes\id^{\otimes k}$ as follows
\begin{align}\label{eq:ConjugatedTraceOutMerminOperatorOur}
\begin{split}
   \tilde{\mathcal{B}}_{N\backslash k}:&= \sqrt{X}_+^{\otimes n}\mathcal{B}_{N\backslash k}\sqrt{X}_+^{\otimes N}\\
    &=\id_{1}\dots\id_{N-k}X_{N-k+1}\dots X_N \\
                 &\quad + Z_{1} Z_{2} \id_{3}\id_{4}\dots \id_{N-k}X_{N-k+1}\dots X_N +   \mbox{ perm.} (1\dots (N-k)) \\
                 &\quad + \dots.\\
                &=\sum_{m\, \mathrm{odd}}Z^{\otimes m}\otimes\id^{\otimes N-k-m}\otimes X^{\otimes k}+\mbox{ perm.}(1\dots (N-k))
\end{split}
\end{align}
Similar to the calculations for the part on nonlocality, we write
 \begin{align*}
\bra{H^3_N}\mathcal{B}_{N\setminus k}\ket{H^3_N}=\bra{\tilde{H}^3_N}\tilde{\mathcal{B}}_{N\backslash k}\ket{\tilde{H}^3_N}=\frac{1}{2}\left(\bra{GHZ}+\bra{GHZ_X^-}\right)\tilde{\mathcal{B}}_{N\setminus k}\left(\ket{GHZ}+\ket{GHZ_X^-}\right)
\end{align*} 
where we also used that we can rewrite the sum over the odd weights as GHZ state in the $X$-basis,
\begin{equation}
\frac{1}{\sqrt{2}}\ket{GHZ_X^-}=\frac{1}{\sqrt{2}^N}\sum_{w(x)\,\mathrm{odd}}\ket{x}=\frac{1}{2}\left(\ket{+}^{\otimes N}-\ket{-}^{\otimes N}\right)
\end{equation}
We compute the values of a single correlator $Z^{\otimes m}\otimes\id^{\otimes N-m-k}\otimes X^{\otimes k}=(Z^{\otimes m}\otimes\id^{\otimes N-m})(\id^{\otimes N-k}\otimes X^{\otimes k})$ on the three different combinations of odd and even-weight contributions. Note that $\ket{GHZ}$ absorbs the $Z^{\otimes m}$-terms since $m$ is even, and the $X^{\otimes k}$-part maps $\ket{GHZ_X^-}$ to $\ket{GHZ_X^\pm}$, depending on $k$. First consider the $GHZ - GHZ$ pairing:
\begin{align*}
\bra{GHZ}Z^{\otimes m}\otimes\id^{\otimes N-m-k}\otimes X^{\otimes k}\ket{GHZ}&=\bra{GHZ}\id_1\dots\id_{N-k} X_{N-k+1}\dots X_N\ket{GHZ}\\
&=0,\\
\intertext{since we assume $N > k > 0$. The cross-terms yield}
\bra{GHZ}Z^{\otimes m}\otimes\id^{\otimes N-m-k}\otimes X^{\otimes k}\ket{GHZ_X^-}&=\begin{cases}
\braket{GHZ}{GHZ_X^-}=0,&k \textup{ even}\\
\braket{GHZ}{GHZ_X^+}=\frac{2}{\sqrt{2}^N},&k\textup{ odd}.
\end{cases}\\
\intertext{And finally,}
\bra{GHZ_X^-}Z^{\otimes m}\otimes\id^{\otimes N-m-k}\otimes X^{\otimes k}\ket{GHZ_X^-}&=\bra{GHZ_X^-}Z_1\dots Z_m \id_{m+1}\dots \id_{N-k}\ket{GHZ_X^\pm}\\
&=\begin{cases}
1 &, k\textup{ is even and }m=0\\
0 &,\textup{ otherwise}.
\end{cases}
\end{align*}
Counting in total $2^{N-k-1}$ of such correlators, we end up with
\begin{equation}
\bra{H^3_N}\mathcal{B}_{N\backslash k}\ket{H^3_N}=\sqrt{2}^{N-2k}
\end{equation}
if $k$ is odd and $\mean{\mathcal{B}_{N\backslash k}}_{\ket{H^3_N}}=\frac{1}{2}$ in the case where $k$ is even.\vspace{2pt}\\
\underline{Case 2: $N\equiv 0\mod 4$}\\
Again, our first step is to transform the inequality:
\begin{align}\label{eq:ConjugatedTraceOutMerminOperatorOur2}
\begin{split}
   \tilde{\mathcal{B}}_{N\backslash k}:=  & \sqrt{Z}_-^{\otimes N}\sqrt{Y}_+^{\otimes N}\mathcal{B}_{N\backslash k}\sqrt{Y}_-^{\otimes N}\sqrt{Z}_-^{\otimes N}\\
    =& (-i)X_{1}\id_{2}\id_{3}\dots \id_{N-k}Z_{N-k+1}\dots Z_N +   \mbox{ perm. } (1\dots (N-k)) \\
                 +& (-i) X_{1}X_{2}X_{3}\id_{4}\id_{5}\dots \id_{N-k}Z_{N-k+1}\dots Z_N +   \mbox{ perm. } (1\dots (N-k)) \\
               + & \dots.\\
                = & (-i)\sum_{m\, \mathrm{odd}} X^{\otimes m}\id^{N-m-k}Z^{\otimes k}+\mbox{perm.}(1\dots (N-k))
\end{split}
\end{align}
Because $N\equiv 0\mod 4$, we have $\bra{\tilde{H}^3_N}\sqrt{Z}_+=\frac{1}{\sqrt{2}}\left(\pm\bra{GHZ}+i\bra{GHZ_X^-}\right)$ and similarly $\sqrt{Z}_+\ket{\tilde{H}^3_N}=\frac{1}{\sqrt{2}}\left(\pm\ket{GHZ}+i\ket{GHZ_X^-}\right)$, see Lemma~\ref{thm:theorem1}. Without loss of generality, let us assume that we have a positive sign in front of the $Z$-basis GHZ part.
All summands of $\tilde{\mathcal{B}}_{N\backslash k}$ interchange between even and odd weights, therefore only the cross-terms contribute: 
\begin{align*}
\bra{H^3_N}\mathcal{B}_{N\setminus k}\ket{H^3_N}=&\frac{1}{2}(\bra{GHZ}+i\bra{GHZ_X^-})\tilde{\mathcal{B}}_{N\backslash k}(\ket{GHZ}+i\ket{GHZ_X^-})\\
=&\bra{GHZ}X^{\otimes m}\id^{N-m-k}Z^{\otimes k}\ket{GHZ_X^-}.
\end{align*}
We let the $Z^{\otimes k}$-part of the operator act on $\bra{GHZ}$ and the $X^{\otimes m}$-components on $\ket{GHZ_X^-}$. Since $m$ is always odd, we have 
\begin{align*}
\bra{GHZ}(\id^{\otimes N-k}\otimes Z^{\otimes k})(X^{\otimes m}\otimes\id^{\otimes N-m})\ket{GHZ_X^-}=\begin{cases}
\braket{GHZ}{GHZ_X^+}=\frac{2}{\sqrt{2}^N},&k\textup{ even},\\
\braket{GHZ^-}{GHZ_X^+}=0, & k \textup{ odd}.
\end{cases}
\end{align*}
Counting in total $2^{N-k-1}$ contributions, we indeed arrive at
\begin{equation}
\bra{H^3_N}\mathcal{B}_{N\backslash k}\ket{H^3_N}=\sqrt{2}^{N-2k}, \textup{ for even }k.
\end{equation}
It is evident that for odd $k$ we get correlation $0$.\\
\underline{Case 3: $N-k\equiv 2\mod 4$}\\
Again, we consider the "odd" Mermin-operator $\mathcal{M}^1_{N-k}$.
We now decompose the hypergraph state, by conditioning on the last $k$ qubits
\begin{align*}
\ket{H^3_N}&=\frac{1}{\sqrt{2}}\ket{H^3_{N-1}}\ket{0}+\frac{1}{\sqrt{2}}\ket{H^{3,2}_{N-1}}\ket{1}\\
&=\frac{1}{2}\left(\ket{H^3_{N-2}}\ket{00}+\ket{H^{3,2}_{N-2}}(\ket{01}+\ket{10})+\ket{H^{3,1}_{N-2}}\ket{11}\right)=\dots\\
&=\frac{1}{\sqrt{2}^k}\sum_{l=0}^k\ket{H_{N-k}(l)}\left(\sum_{w(x)=l}\ket{x}\right)
\end{align*}
Since $\ket{H^{3,1}_m}=\frac{1}{\sqrt{2}}\left(\ket{H^{3,1}_{m-1}}\ket{0}-\ket{H^{3,2,1}_{m-1}}\ket{1}\right)$ and $\ket{H^{3,2,1}_m}=\frac{1}{\sqrt{2}}\left(\ket{H^{3,2,1}_{m-1}}\ket{0}-\ket{H^{3}_{m-1}}\ket{1}\right)$, the states $\ket{H_{N-k}(l)}$ only depend on $l\mod 4$, more specifically
\begin{equation}
\ket{H_{N-k}(l)}=\begin{cases}
\hphantom{-}\ket{H^3_{N-k}},&l\equiv 0\mod 4\\
\hphantom{-}\ket{H^{3,2}_{N-k}},& l\equiv 1\mod 4\\
\hphantom{-}\ket{H^{3,1}_{N-k}},& l\equiv 2\mod 4\\
-\ket{H^{3,2,1}_{N-k}},\, & l\equiv 3\mod 4.
\end{cases}
\end{equation}
With this, we decompose the correlator by measuring on the last qubits, to obtain mixed-state overlaps:
\begin{align}
\nonumber\bra{H^3_N}\mathcal{B}_{N\setminus k}\ket{H^3_N}=&\frac{1}{2^k}\sum_{l=0}^k \bra{H_{N-k}(l)}\left(\sum_{w(x)=l}\bra{x}\right)\mathcal{M}^1_{N-k}\otimes \id^k\sum_{l'=0}^k\ket{H_{N-k}(l')}\left(\sum_{w(x)=l'}\ket{x'}\right)\\
\nonumber=&\frac{1}{2^k}\sum_{l,l'=0}^k\sum_{\stackrel{w(x)=l}{w(x')=l'}} \bra{H_{N-k}(l)}\mathcal{M}^1_{N-k}\ket{H_{N-k}(l')}\braket{x}{x'}\\
\label{eq:aux13}=&\frac{1}{2^k}\sum_{l=0}^k\binom{k}{l}\bra{H_{N-k}(l)}\mathcal{M}^1_{N-k}\ket{H_{N-k}(l)}.
\end{align}
By construction $\ket{H^3_{N-k}}$ and $\ket{H^{3,1}_{N-k}}$ are $X^{\otimes N-k}$-stabilized and from Lemma \ref{lemma:localStab} it is evident that $\ket{H^{3,2}_{N-k}},\ket{H^{3,2,1}_{N-k}}$ are invariant under $Y^{\otimes N-k}$. In order to calculate the contributions originating from the $X$-stabilized cases, corresponding to $l\mod 4=0,2$ we transform the inequality to
\begin{align*}
\tilde{\mathcal{B}}^1_{N-k}:=&\sqrt{X}_+^{\otimes N-k}\mathcal{M}^1_{N-k}\sqrt{X}_+^{\otimes N-k}\\
=&\sum_{m\,\mathrm{odd}}iZ^{\otimes m}\otimes\id^{\otimes N-k-m} +\mbox{perm}
 \end{align*}
For the $Y$-stabilized cases $l\equiv 1,3\mod 4$, we consider
\begin{align*}
\mathcal{B}'_{N-k}:=&\sqrt{Z}_-^{\otimes N-k}\sqrt{Y}_+^{\otimes N-k} \mathcal{M}^1_{N-k} \sqrt{Y}_-^{\otimes N}\sqrt{Z}_-^{\otimes N}\\
=&\sum_{m\, \mathrm{odd}} i X^{\otimes m}\otimes \id^{N-m-k}+\mbox{perm.}
\end{align*}
First consider $l\equiv 0\mod 4$
\begin{align}
\bra{\tilde{H}^3_{N-k}}\tilde{\mathcal{B}}^1_{N-k}\ket{\tilde{H}^3_N}=\frac{1}{2}\left(\bra{GHZ}+\bra{GHZ_X^-}\right)\tilde{\mathcal{B}}_{N-k}\left(\ket{GHZ}+\ket{GHZ_X^-}\right)
\end{align}
We compute the contributions separately. Because $m$ is odd, we have
\begin{align*}
\bra{GHZ}Z_1\dots Z_m \id_{m+1}\dots \id_{N-k}\ket{GHZ}&=\braket{GHZ^-}{GHZ}=0\\
\bra{GHZ}Z_1\dots Z_m \id_{m+1}\dots \id_{N-k}\ket{GHZ_X^-}&=\braket{GHZ^-}{GHZ_X^-}=0\\
\bra{GHZ_X^-}Z_1\dots Z_m \id_{m+1}\dots \id_{N-k}\ket{GHZ_X^-}&=0,
\end{align*}
since $Z$ interchanges between $\ket{+}$ and $\ket{-}$. Therfore
\begin{equation}
\bra{H^3_{N-k}}\mathcal{B}_{N-k}\ket{H^3_{N-k}}=\bra{H^{3,1}_{N-k}}\mathcal{B}_{N-k}\ket{H^{3,1}_{N-k}}=0.
\end{equation}
Next we look at the case where $l\equiv 1\mod 4$.\\
We rewrite the correlator using the latter transformation of the Bell operator:
\begin{align*}
\bra{H^{3,2}_{N-k}}\mathcal{B}_{N-k}\ket{H^{3,2}_{N-k}}=\bra{H^{3,2}_{N-k}}\sqrt{Y}_-^{\otimes N-k}\sqrt{Z}_+^{\otimes N-k}{\mathcal{B}'}_{N-k}\sqrt{Z}_+^{\otimes N-k}\sqrt{Y}_+^{\otimes N-k}\ket{H^{3,2}_{N-k}}
\end{align*}
With the aid of Lemma~\ref{lemma:3,2-uniform}, we can transform the state as
\begin{equation*}
\sqrt{Y}_+^{\otimes N}\ket{H^{3,2}_{N-k}}=\pm\frac{i}{\sqrt{2}}\ket{GHZ^-}+\frac{1}{\sqrt{2}}\ket{GHZ_Y^-}.
\end{equation*}
Without loss of generality, we shall assume the first sign to be negative. Applying $\sqrt{Z}_+^{\otimes N}$ to the respective ket-. and bra-states yields
\begin{equation*}
\sqrt{Z}_+^{\otimes N-k}\sqrt{Y}_+^{\otimes N-k}\ket{H^{3,2}_{N-k}}=\pm\frac{i}{\sqrt{2}}\ket{GHZ}+\frac{1}{\sqrt{2}}\ket{GHZ_X^-},
\end{equation*}
and
\begin{equation*}
\bra{H^{3,2}_{N-k}}\sqrt{Y}_-^{\otimes N-k}\sqrt{Z}_+^{\otimes N-k}=\mp\frac{i}{\sqrt{2}}\bra{GHZ}-\frac{1}{\sqrt{2}}\bra{GHZ_X^-}.
\end{equation*}
Note that this again requires $N\equiv 2\mod 4$ to map $\ket{GHZ^-}$ to $\ket{GHZ}$. Then
\begin{align*}
\bra{GHZ^+}X_1\dots X_m \id_{m+1}\dots \id_{N-k}\ket{GHZ^+}&=0\\
\bra{GHZ^+}X_1\dots X_m \id_{m+1}\dots \id_{N-k}\ket{GHZ_X^-}&=\braket{GHZ^+}{GHZ_X^+}=\frac{2}{\sqrt{2}^{N-k}}\\
\bra{GHZ_X^-}X_1\dots X_m \id_{m+1}\dots \id_{N-k}\ket{GHZ_X^-}&=\braket{GHZ_X^+}{GHZ_X^-}=0
\end{align*}
Therefore, when decomposing the quantum value into the different pairings, only the cross-terms yield a contribution and we are left with
\begin{align*}
\bra{H^{3,2}_{N-k}}\mathcal{M}^1_{N-k}\ket{H^{3,2}_{N-k}}&=\frac{1}{2}\left(i\bra{GHZ^+}-\bra{GHZ_X^-}\right)\mathcal{B}'_{N-k}\left(-i\ket{GHZ_Z^+}+\ket{GHZ_X^-}\right)\\
&=i\bra{GHZ^+}\mathcal{B}_{N-k}'\ket{GHZ_X^-}\\
&=\sum_{m\, \mathrm{odd}}\binom{N-k}{m}\bra{GHZ^+}X_1\dots X_m\id_{m+1}\dots\id_{N-k}\ket{GHZ_X^-}\\
&=2^{N-k-1}\frac{2}{\sqrt{2}^{N-k}}=\sqrt{2}^{N-k}.
\end{align*}
Because conjugation with $Z^{\otimes N-k}$ flips the sign of $\mathcal{M}^1_{N-k}$, the $3,2,1$-uniform state has the same contribution with opposite sign:
\begin{equation*}
\bra{H^{3,2,1}_{N-k}}\mathcal{M}^1_{N-k}\ket{H^{3,2,1}_{N-k}}=-\bra{H^{3,2}_{N-k}}\mathcal{M}^1_{N-k}\ket{H^{3,2}_{N-k}}=\sqrt{2}^{N-k}.
\end{equation*}
Hence, whenever $k\geq 1$, \eqref{eq:aux13} becomes
\begin{align}
\label{eq:aux14}\bra{H^3_N}\mathcal{B}_{N\setminus k}\ket{H^3_N}=&\frac{(-i)}{2^k}\sum_{l=1,\mathrm{odd}}^k\binom{k}{l}i^l\sqrt{2}^{N-k}\\
=&(-i)\sqrt{2}^{N-3k}\frac{1}{2}\left((1+i)^k-(1-i)^k\right)\\
\label{eq:final}=&\sin\left(\frac{\pi k}{4}\right)\sqrt{2}^{N-2k},
\end{align}
which can be made nonnegative by choosing the sign of the inequality appropriately in the beginning.\\
\underline{Case 4: $N-k\equiv 0\mod 4$}\\
As far as the last claim is concerned, the derivation is identical of the previous one until Eq.~\eqref{eq:aux13}.
The difference here is, that $\ket{H^3_{N-k}}$ is $Y$-stabilized, whereas $\ket{H_{N-k}^{3,2}}$  is $X$-stabilized. With the same tricks as before, we can express
\begin{align}
    \bra{H^{3}_{N-k}}\mathcal{M}_{N-k}^1\ket{H^3_{N-k}}&=\frac{1}{2}\left(\pm\bra{GHZ}+i\bra{GHZ_X^-}\right)\mathcal{B}'\left(\pm\ket{GHZ}+i\ket{GHZ_X^-}\right)\\
    &=\pm\braket{GHZ}{GHZ_X^+}=\pm\frac{1}{\sqrt{2}^{N-k}}.
\end{align}
Since $\mathcal{B}'$ consists only of summands with an odd number of Pauli-$X$ and identities otherwise, the cross-terms give the only contribution, where $\ket{GHZ_X^-}$ gets mapped to $\ket{GHZ_X^+}$.
According to Lemma~\ref{lemma:3,2-uniform}, we have $\sqrt{X}^{\otimes N}\ket{H^{3,2}_N}=\pm\frac{1}{\sqrt{2}}\ket{GHZ}+\frac{1}{\sqrt{2}}\ket{GHZ_X^-}$. Similarly as before, we can rewrite
\begin{align}
    \bra{H^{3,2}_{N-k}}\mathcal{M}^1_{N-k}\ket{H^{3,2}_{N-k}}=\frac{1}{2}(\bra{GHZ}+\bra{GHZ_X^-})\tilde{\mathcal{B}}_{N-k}(\ket{GHZ}+\ket{GHZ_X^-})
\end{align}
It is now straightforward to check that all the four contributions evaluate to zero, as $\tilde{\mathcal{B}}$ only features odd number of Pauli $Z$ and identities otherwise.
The rest of the proof then follows analogously, with the exception of the alternating $l$-summation in Eq.~\eqref{eq:aux14} running over even integers instead, so that in Eq.~\eqref{eq:final}, the sine gets replaced by a cosine, which concludes the proof.
\end{proof}
\section*{References}
\bibliography{iopart-num}
\end{document}